\documentclass{article}

\PassOptionsToPackage{square,numbers,sort&compress}{natbib}

\usepackage[preprint]{neurips_2021}

\usepackage[utf8]{inputenc} 
\usepackage[T1]{fontenc}    
\usepackage[hidelinks]{hyperref}       
\usepackage{url}            
\usepackage{booktabs,multirow}       
\usepackage{amsfonts}       
\usepackage{nicefrac}       
\usepackage{microtype}      
\usepackage{xcolor}         

\usepackage{amsmath,mathtools,amsthm}
\usepackage{color}
\usepackage[ruled,linesnumbered]{algorithm2e}
\usepackage{algorithmicx}
\usepackage{array}
\usepackage{rotating}
\usepackage{xcolor}
\usepackage{subfigure}
\usepackage{enumitem}
\usepackage{placeins}

\usepackage{xr}
\makeatletter
\newcommand*{\addFileDependency}[1]{
  \typeout{(#1)}
  \@addtofilelist{#1}
  \IfFileExists{#1}{}{\typeout{No file #1.}}
}
\makeatother

\newcommand*{\myexternaldocument}[1]{
    \externaldocument{#1}
    \addFileDependency{#1.tex}
    \addFileDependency{#1.aux}
}
\myexternaldocument{supplement}

\usepackage{mathtools}
\DeclarePairedDelimiter{\ceil}{\lceil}{\rceil}
\newcommand{\defeq}{\vcentcolon=}

\newcommand{\indep}{\;\rotatebox[origin=c]{90}{$\models$}\;}

\renewcommand{\P}[1]{\mathbb{P}{\left[#1\right]}}
\newcommand{\E}[1]{\mathbb{E}{\left[#1\right]}}

\usepackage{bbm}
\newcommand{\I}[1]{\mathbbm{1}\left[#1\right]}

\newtheorem{theorem}{Theorem}
\newtheorem{proposition}{Proposition}
\newtheorem{lemma}{Lemma}
\newtheorem{assumption}{Assumption}

\title{Conformal Prediction using Conditional Histograms}

\author{%
  Matteo Sesia\\
  Department of Data Sciences and Operations \\
  University of Southern California, USA \\
  \texttt{sesia@marshall.usc.edu}
  \And
  Yaniv Romano \\
  Departments of Electrical and of Computer Engineering\\
  and of Computer Science \\
  Technion, Israel \\
  \texttt{yromano@technion.ac.il}
}

\begin{document}

\maketitle

\begin{abstract}
This paper develops a conformal method to compute prediction intervals for non-parametric regression that can automatically adapt to skewed data.
Leveraging black-box machine learning algorithms to estimate the conditional distribution of the outcome using histograms, it translates their output into the shortest prediction intervals with approximate conditional coverage.
The resulting prediction intervals provably have marginal coverage in finite samples, while asymptotically achieving conditional coverage and optimal length if the black-box model is consistent. Numerical experiments with simulated and real data demonstrate improved performance compared to state-of-the-art alternatives, including conformalized quantile regression and other distributional conformal prediction approaches.
\end{abstract}

\section{Introduction}

\subsection{Problem statement and motivation} \label{sec:intro-statement}

We consider the problem of predicting {\em with confidence} a response variable $Y \in \mathbb{R}$ given $p$ features $X \in \mathbb{R}^p$  for a test point $n+1$, utilizing $n$ pairs of observations $\{(X^{(i)}, Y^{(i)})\}_{i=1}^{n}$ drawn {\em exchangeably} (e.g., i.i.d.) from some unknown distribution, and leveraging any machine-learning algorithm.
Precisely, $\forall \alpha \in (0,1)$, we seek a prediction {\em interval} $\hat{C}_{n,\alpha}(X_{n+1}) \subset \mathbb{R}$ for $Y_{n+1}$ satisfying the following three criteria. First, $\hat{C}_{n,\alpha}$ should have finite-sample marginal coverage at level $1-\alpha$,
\begin{align} \label{eq:marginal-coverage}
    \mathbb{P}\left[Y_{n+1} \in \hat{C}_{n,\alpha}(X_{n+1}) \right] \geq 1-\alpha.
\end{align}
Second, $\hat{C}_{n,\alpha}$ should {\em approximately} have conditional coverage at level $1-\alpha$,
\begin{align} \label{eq:conditional-coverage}
    \mathbb{P}\left[ Y_{n+1} \in \hat{C}_{n,\alpha}(x) \mid X_{n+1} = x \right] \geq 1-\alpha,  \qquad \forall x \in \mathbb{R}^p,
\end{align}
meaning it should approximate this objective in practice, and ideally achieve it asymptotically under suitable conditions in the limit of large sample sizes. Third, $\hat{C}_{n,\alpha}$ should be as narrow as possible.

We tackle this challenge with conformal inference~\cite{vovk2005algorithmic,lei2018distribution}, which allows one to convert the output of any black-box machine learning algorithm into prediction intervals with provable marginal coverage~\eqref{eq:marginal-coverage}.
The key idea of this framework is to compute a {\em conformity score} for each observation, measuring the discrepancy, according to some metric, between the true value of $Y$ and that predicted by the black-box model. The model fitted on the training data is then applied to hold-out calibration samples, producing a collection of conformity scores. As all data points are exchangeable, the empirical distribution of the calibration scores can be leveraged to make predictive inferences about the conformity score of a new test point. Finally, inverting the function defining the conformity scores yields a prediction set for the test $Y$.
This framework can accommodate almost any choice of conformity scores, and in fact many different implementations have already been proposed to address our problem.
However, it remains unclear how to implement a concrete method from this broad family that can lead to the most informative possible prediction intervals.
Our contribution here is to develop a practical solution, following the three criteria defined above, that performs better compared to existing alternatives and is asymptotically optimal under certain assumptions.

It is worth emphasizing that constructing a short prediction interval with guaranteed coverage is a reasonable approach to quantify and communicate predictive uncertainty in regression problems, although it is of course not the only one.
To name an alternative, one could compute a non-convex prediction set with analogous coverage \cite{izbicki2020cd}, which might be more appropriate in some situations, but is also more easily confusing. 
For example, it could be informative for a physician to know that the future
blood pressure of a patient with certain characteristics is predicted to be within the range
[120,129] mmHg. However, it would not be more helpful to report instead the following non-convex region: $[120, 120.012] \cup [120.015, 120.05] \cup
[121, 122.7] \cup [123.1, 127.2] \cup [127.8, 129]$ mmHg.
Indeed, in the second case it would not be clear (a) whether the multi-modal nature of that prediction is significant or a spurious consequence of overfitting, and (b) how the
physician would act upon that prediction any differently than if it had been [120,129] mmHg. 
Therefore, we focus on prediction intervals in this paper because they are generally easier to interpret than arbitrary regions, and they are also less likely to convey a false sense of confidence.

\subsection{Preview of conformal histogram regression}

Imagine an {\em oracle} with access to $P_{Y \mid X}$, the distribution of $Y$ conditional on $X$, which leverages such information to construct optimal prediction intervals as follows. For simplicity, suppose $P_{Y \mid X}$ has a continuous density $f(y \mid x)$ with respect to the Lebesgue measure, although this could be relaxed with more involved notation. Then, the oracle interval for $Y_{n+1} \mid X_{n+1}=x$ would be:
\begin{align} \label{eq:oracle-interval}
  C^{\mathrm{oracle}}_{\alpha}(x) & = \left[ l_{1-\alpha}^{\mathrm{oracle}}(x), u_{1-\alpha}^{\mathrm{oracle}}(x) \right],
\end{align}
where, for any $\tau \in (0,1]$, $l_{\tau}^{\mathrm{oracle}}(x)$ and $u_{\tau}^{\mathrm{oracle}}(x)$ are defined as:
\begin{align} \label{eq:oracle-interval-int}
  [l_{\tau}^{\mathrm{oracle}}(x), u_{\tau}^{\mathrm{oracle}}(x)] & \defeq \mathop{\mathrm{arg\,min}}_{(l,u) \in \mathbb{R}^2 \, : \, l \leq u}  \left\{ u-l : \int_{l}^{u} f(y \mid x) dy \geq \tau \right\}.
\end{align}
This is the shortest interval with conditional coverage~\eqref{eq:conditional-coverage}.
If the solution to~\eqref{eq:oracle-interval-int} is not unique (e.g., if $f(\cdot \mid x)$ is piece-wise constant), the oracle picks any solution at random.
Of course, this is not a practical method because $f$ is unknown.
Therefore, we will approximate~\eqref{eq:oracle-interval-int} by fitting a black-box model on the training data, and then use conformal prediction to construct an interval accounting for any possible estimation errors.
Specifically, we replace $f$ in~\eqref{eq:oracle-interval-int} with a histogram approximation, hence why we call our method {\em conformal histogram regression}, or CHR. The output interval is then
\begin{align} \label{eq:conformal-interval}
  \hat{C}_{n,\alpha}(x) & = \left[ \hat{l}_{\hat{\tau}}(x), \hat{u}_{\hat{\tau}}(x) \right],
\end{align}
where $\hat{l}_{\hat{\tau}}(x)$ and $\hat{u}_{\hat{\tau}}(x)$ approximate the analogous oracle quantities in~\eqref{eq:oracle-interval-int}. The value of $\hat{\tau}$ in~\eqref{eq:conformal-interval} will be determined by suitable conformity scores evaluated on the hold-out data, and it may be larger than $1-\alpha$ if the model for $f$ is not very accurate. However, if the fitted histogram is close to the true $P_{Y \mid X}$,
the interval in~\eqref{eq:conformal-interval} will resemble that of the oracle~\eqref{eq:oracle-interval}.

Figure~\ref{fig:example-1} previews an application to toy data, comparing CHR to conformalized quantile regression (CQR)~\cite{romano2019conformalized}; see Section~\ref{sec:exp-synthetic} for more details. CHR finds the shortest interval such the corresponding area under the histogram is above $\tau$, for any $\tau \in (0,1]$, and then calibrates $\tau$ to guarantee marginal coverage above $1-\alpha$; this extracts more information from the model compared to CQR. 
For example, CHR adapts automatically to the skewness of $Y \mid X$, returning intervals delimited by the 0\%--90\% quantiles in this example, which are shorter than the symmetric ones (5\%--95\%) sought by CQR.

\begin{figure}[!htb]
  \centering
  \begin{minipage}[t]{0.03\textwidth}
    (a)
  \end{minipage}
  \begin{minipage}[]{0.45\textwidth}
    \includegraphics[width=\textwidth]{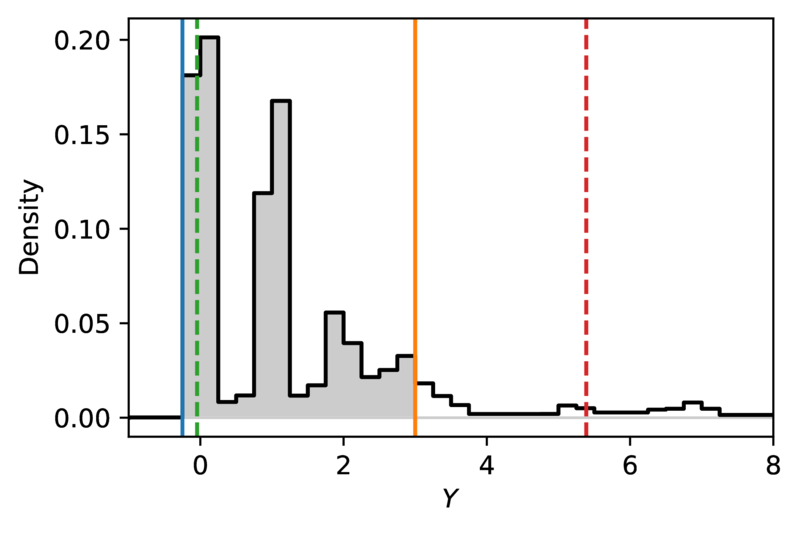}
  \end{minipage}
  \begin{minipage}[t]{0.03\textwidth}
    (b)
  \end{minipage}
  \begin{minipage}[]{0.45\textwidth}
    \includegraphics[width=\textwidth]{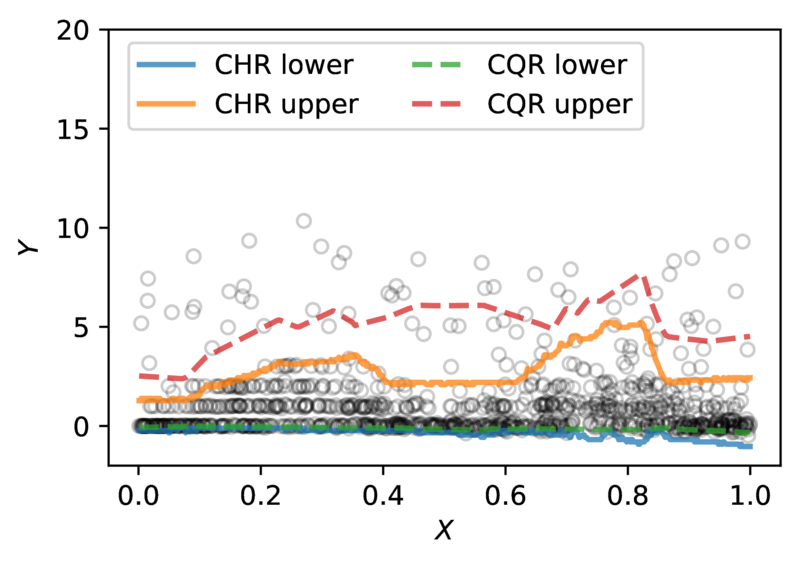}
  \end{minipage}
\caption{CHR prediction intervals in an example with one variable, compared to those obtained with CQR~\cite{romano2019conformalized}. Both methods guarantee 90\% marginal coverage and are based on the same deep quantile model.
(a)~Histogram estimate of $P_{Y \mid X}$ for a point with $X \approx 0.2$. The CHR interval corresponds to the shaded part of the histogram, whose area is approximately 0.9, as marked by the solid vertical lines. The dashed lines denote the CQR interval. 
(b)~Prediction bands for the two methods, as a function of $X$. CHR: empirical marginal coverage 0.9, estimated conditional coverage 0.9, and average length 3.2. The corresponding quantities for CQR are: 0.9, 0.9, and 5.2, respectively.} 
  \label{fig:example-1}
\end{figure}

\subsection{Related work}

This work is inspired by the conformity scores introduced by~\cite{romano2020classification} for multi-class classification, the underlying idea of which can be repurposed here. Nonetheless, the extension to our problem involves several innovations. This paper connects~\cite{romano2020classification} to other conformal methods for continuous responses~\cite{vovk2005algorithmic,lei2014distribution,lei2018distribution}, which sought objectives similar to ours by leveraging quantile regression~\cite{romano2019conformalized,kivaranovic2020adaptive,sesia2020comparison,gupta2019nested,yang2021finite} or non-parametric density estimation \cite{izbicki2019flexible,chernozhukov2019distributional}, sometimes considering multi-modal prediction sets instead of intervals~\cite{izbicki2020cd}. Our approach also exploits black-box models for distributional estimation; however, we introduce more efficient conformity scores.

We seek the shortest intervals with marginal coverage while approximating as well as possible conditional coverage, although the latter is impossible to guarantee in finite samples~\cite{vovk2012conditional,barber2019limits}.
The performances of prior approaches have been measured in terms of these criteria, yet others have not sought them as directly. 
Indeed, if the black-box model is consistent for $P_{Y \mid X}$, our method becomes asymptotically equivalent to the oracle~\eqref{eq:oracle-interval}--\eqref{eq:oracle-interval-int}, under some technical assumptions. This property does not hold for other existing methods because they tend to produce symmetric intervals, with fixed lower and upper miscoverage rates (the probabilities of the outcome being either below or above the output interval, respectively), which may be sub-optimal if the data have unknown skewness.


\section{Methods} \label{sec:methods}

The proposed method consists of four main components: the estimation and binning of a conditional model for the outcome, the construction of a nested sequence of approximate oracle intervals based on the above, the computation of suitable conformity scores, and their conformal calibration.

\subsection{Estimating conditional histograms}

We partition the domain of $Y$ into $m$ bins $[b_{j-1}, b_{j})$, for some sequence $b_0 < \ldots < b_m$. With little loss of generality, assume $Y$ is bounded: $-C = b_0 < Y < b_m = C$, for some $C>0$. Then, we solve a discrete version of the problem stated in the introduction: we seek the smallest possible contiguous subset of bins with $1-\alpha$ predictive coverage. If $m$ is large and the bins are narrow, this problem is not very different from the original one, although it is more amenable to solution.

For simplicity, we present our method from a split-conformal perspective~\cite{vovk2005algorithmic,lei2018distribution}; extensions to other hold-out approaches~\cite{vovk2005algorithmic,barber2019predictive,kim2020predictive} will be intuitive.
Let $\mathcal{D}^{\mathrm{train}}, \mathcal{D}^{\mathrm{cal}} \subset \{1,\ldots,n\}$ denote any partition of the data into training and calibration subsets, respectively.
$\mathcal{D}^{\mathrm{train}}$ is used to train a black-box model for the conditional probabilities that $Y$ is within any of the above bins: $\forall j \in \{1,\ldots,m\}$,
\begin{align} \label{eq:def-pi}
  \pi_j(x) \defeq \P{ Y \in [b_{j-1}, b_j)  \mid X = x}.
\end{align}
There exist many tools to approximate $P_{Y \mid X}$ and obtain estimates $\hat{\pi}_j(x)$ of $\pi_j(x)$, including quantile regression \cite{taylor2000quantile,meinshausen2006quantile,moon2021learning}, Bayesian additive regression trees~\cite{chipman2010bart}, or any other non-parametric conditional density estimator \cite{magdon1999neural,izbicki2016nonparametric,dalmasso2020conditional}. Our method can directly be applied with any of these models, but we found multiple quantile regression to work particularly well~\cite{romano2019conformalized,kivaranovic2020adaptive,sesia2020comparison,gupta2019nested}, and therefore we will focus on it in this paper. Referring to Supplementary Section~\ref{sec:supp-histograms} for implementation details and information about the computational cost of the learning algorithm (which is comparable to that required by CQR~\cite{romano2019conformalized}), we thus take these black-box estimates $\hat{\pi}_j(x)$ as fixed henceforth.

Note that estimating conditional distributions is more challenging if the number of variables is larger. However, this is a fundamental difficulty of high-dimensional regression, not a particular limitation of the proposed CHR.
Although our method utilizes conditional histograms learnt from the data, its performance is not directly measured in terms of how closely these resemble the true $P_{Y \mid X}$. Instead, as we shall see, CHR only needs to detect the possible skewness of $Y \mid X$ and estimate reasonably well some lower and upper quantiles of this conditional distribution. Therefore, its estimation task is not much more difficult than that of CQR~\cite{romano2019conformalized}, as skewness is relatively easy to detect.

\subsection{Constructing a nested sequence of approximate oracle intervals}

For any partition $\mathcal{B} = (b_0, \ldots, b_m)$ of the domain of $Y$, let $\pi = (\pi_1, \ldots, \pi_m)$ be a unit-sum sequence, depending on $x \in \mathbb{R}^p$; this may be seen as a histogram approximation of $P_{Y \mid X}$~\eqref{eq:def-pi}.
For simplicity, assume all histogram bins have equal width, although this is unnecessary.
Then, define the following bi-valued function $\mathcal{S}$ taking as input $x \in \mathbb{R}^p$, $\pi$, $\tau \in (0,1]$, and two intervals $S^-, S^+ \subseteq \{1,\ldots,m\}$:
\begin{align} \label{eq:def-opt-problem}
  \mathcal{S}(x, \pi, S^-, S^+, \tau) \defeq \mathop{\mathrm{arg\,min}}_{(l,u) \in \{1,\ldots,m\}^2 \,:\, l \leq u} \left\{ |u-l| : \sum_{j=l}^{u} \pi_j(x) \geq \tau, \,  S^- \subseteq [l,u] \subseteq S^+  \right\}.
\end{align}
Above, it is implicitly understood we choose the value of $(l,u)$ minimizing $\sum_{j=l}^{u} \pi_j(x)$ among the feasible ones with minimal $|u-l|$, if the optimal solution would not otherwise be unique. 
Therefore, we can assume without loss of generality the solution to~\eqref{eq:def-opt-problem} is unique; if that is not the case, we can break the ties at random by adding a little noise to $\pi$.
The problem in~\eqref{eq:def-opt-problem} can be solved at computational cost linear in the number of bins, and it is equivalent to the standard programming challenge of finding the smallest positive subarray whose sum is above a given threshold.
Note that we will sometimes refer to intervals on the grid determined by $\mathcal{B}$ as either contiguous subsets of $\{1,\ldots,m\}$ (e.g., $S^-$) or as pairs of lower and upper endpoints (e.g., $[l,u]$).

If $S^- = \emptyset$ and $S^+ = \{1,\ldots,m\}$, the expression in~\eqref{eq:def-opt-problem} computes the shortest possible interval with total mass above $\tau$ according to $\pi(x)$.
Further, if $\pi_j$ is the mass in the $j$-th bin according to the true $P_{Y \mid X}$, then $\mathcal{S}(x, \pi, \emptyset, \{1,\ldots,m\}, 1-\alpha)$ is the discretized version of the oracle interval~\eqref{eq:oracle-interval}--\eqref{eq:oracle-interval-int}.
In general, the optimization in~\eqref{eq:def-opt-problem} involves the additional {\em nesting} constraint that the output $\mathcal{S}$ must satisfy $S^- \subseteq \mathcal{S} \subseteq S^+$, which will be needed to guarantee our method has valid marginal coverage~\cite{gupta2019nested}.
Intuitively, it is helpful to work with a nested sequence because this ensures the prediction intervals are monotone increasing in $\tau$, essentially reducing the calibration problem to that of selecting the appropriate value of $\tau$ that yields the desired marginal coverage.
Note that the inequality in~\eqref{eq:def-opt-problem} involving $\tau$ may not be binding at the optimal solution due to the discrete nature of the optimization problem. 
However, the above oracle can be easily modified by introducing some suitable randomization in order to obtain valid prediction intervals that are even tighter on average, as explained in Supplementary Section~\ref{sec:supp-random}.

As $\hat{\pi}$ may be an inaccurate estimate of $P_{Y \mid X}$, we cannot simply plug it into the oracle in~\eqref{eq:def-opt-problem} and expect valid coverage. However, for any approximate conditional histogram $\hat{\pi}$, we can define a {\em nested} sequence~\cite{gupta2019nested} of (randomized) sub-intervals of $\mathcal{B}$, for different values of $\tau$ ranging from 0 to 1. Then, we calibrate $\tau$ to obtain the desired $1-\alpha$ marginal coverage.
Precisely, consider an increasing scalar sequence $\tau_t = t/T$, for $t \in \{0,\ldots,T\}$ with some $T \in \mathbb{N}$, and define a corresponding growing sequence of subsets $S_{t} \subseteq \{1,\ldots,m\}$ as follows. First, fix any {\em starting point} $\bar{t} \in \{0,\ldots,T\}$ and define $S_{\bar{t}}$ by applying \eqref{eq:def-opt-problem} without the nesting constraints (with $S^- = \emptyset$ and $S^+ = \{1,\ldots,m\}$):
\begin{align} \label{eq:def-s-start}
  & S_{\bar{t}}  \defeq \mathcal{S}(x, \pi, \emptyset, \{1,\ldots,m\}, \tau_{\bar{t}}),
\end{align}
Note the explicit dependence on $x$ and $\pi$ of the left-hand-side above is omitted for simplicity, although it is important to keep in mind that $S_{\bar{t}}$ does of course depend on these quantities. Figure~\ref{fig:intervals-schematic} (second row) visualizes the construction of $S_{\bar{t}}$ in a toy example with $\tau_{\bar{t}} = 0.9$.

\begin{figure}[!htb]
  \centering
  \includegraphics[width=0.6\textwidth]{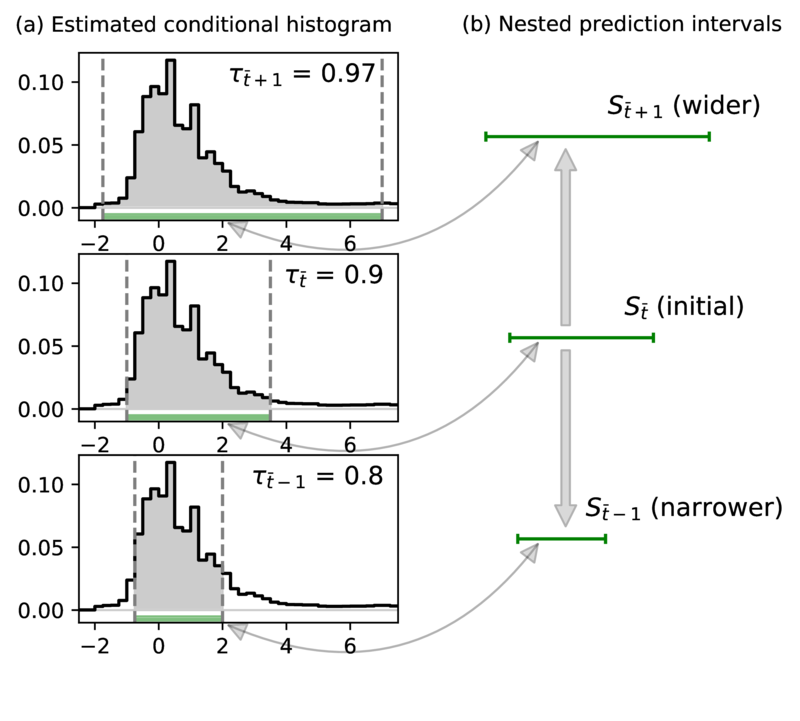}
\caption{Schematics for the construction of a nested sequence of approximate oracle prediction intervals~\eqref{eq:def-s-start}--\eqref{eq:def-s-dec}. (a) Conditional histogram approximation of the distribution of $Y \mid X$, based on a black-box model. The shaded areas delimited by the dashed vertical lines denote the shortest intervals with the desired mass ($\tau$) under the histograms, subject to the nesting constraints. (b) Sequence of prediction intervals. The initial interval $S_{\bar{t}}$ is not subject to any nesting constraints. The wider (above), or narrower (below), intervals must contain $S_{\bar{t}}$ (above), or be contained in it (below).}
  \label{fig:intervals-schematic}
\end{figure}

Having computed the initial interval $S_{t}$ for $t = \bar{t}$, we recursively extend the definition to the wider intervals indexed by $t = \bar{t} + 1, \ldots, T$ as follows:
\begin{align} \label{eq:def-s-inc}
  S_{t} & \defeq \mathcal{S}(x, \pi, S_{t-1}, \{1,\ldots,m\}, \tau_{t}).
\end{align}
See the top row of Figure~\ref{fig:intervals-schematic} for a schematic of this step.
Similarly, the narrower intervals $S_{t}$ indexed by $t = \bar{t}-1, \bar{t}-2, \ldots 0$ are defined recursively as:
\begin{align} \label{eq:def-s-dec}
  & S_{t} \defeq \mathcal{S}(x, \pi, \emptyset, S_{t+1}, \tau_{t}).
\end{align}
See the bottom row of Figure~\ref{fig:intervals-schematic} for a schematic of this step.
As a result of this construction, the sequence of intervals $\{ S_{t} \}_{t=0}^{T}$ is nested regardless of the starting point $\bar{t}$ in~\eqref{eq:def-s-start}, as previewed in Figure~\ref{fig:intervals-schematic}.
However, different choices of $\bar{t}$ may lead to different sequences, any of which allows us to obtain provable marginal coverage, as discussed next. As our goal is to approximate the oracle in~\eqref{eq:oracle-interval}--\eqref{eq:oracle-interval-int} accurately, the most intuitive choice is to pick $\bar{t}$ such that $\tau_{\bar{t}} \approx 1-\alpha$. A more involved randomized version of this construction, inspired by the more powerful randomized oracle, is discussed in Supplementary Section~\ref{sec:supp-random}. Note that the randomized version of the nested prediction intervals will be the one applied throughout this paper and, with a slight overload of notation, we will simply refer to it as $\{S_t\}_{t=1}^{T}$.
Note also that we will highlight the dependence of this sequence on $x$ and $\pi$ by writing it as $S_t(x,\pi)$.
Further, as we work henceforth with the randomized versions of these prediction intervals (described in Supplementary Section~\ref{sec:supp-random}), we will refer to them as $S_t(x,\varepsilon,\pi)$, where $\varepsilon$ is a uniform random variable in $[0,1]$, independent of everything else.

\subsection{Computing conformity scores and calibrating prediction intervals}

Given any sequence of nested sets $S_{t}(x, \varepsilon, \pi)$, we define the following conformity score function $E$:
\begin{align} \label{eq:conformity-function}
  E(x,y,\varepsilon,\pi) \defeq \min \left\{ t \in \{0,\ldots,T\}: y \in S_t(x, \varepsilon, \pi) \right\}.
\end{align}
In words, this computes the smallest index $t$ such that $S_t(x, \varepsilon, \pi)$ contains $y$, as in~\cite{romano2020classification,gupta2019nested}.
Equivalently, one can think of these scores as indicating the smallest value of the nominal coverage $\tau$ in~\eqref{eq:def-opt-problem} necessary to ensure the observed $Y$ is contained in the prediction interval.
Our method evaluates~\eqref{eq:conformity-function} on all calibration samples $(X_i,Y_i)$ using the $\hat{\pi}$ learnt on the training data; for each $i \in \mathcal{D}^{\mathrm{cal}}$, we generate $\varepsilon_i \sim \text{Unif}(0,1)$ and store
\begin{align*}
  E_i = E(X_i, Y_i, \varepsilon_i, \hat{\pi}).
\end{align*}
Then, we compute prediction intervals for $Y_{n+1}$ by looking at the nested sequence in~\eqref{eq:def-s-start}--\eqref{eq:def-s-dec} corresponding to the new $X_{n+1}$ and selecting the interval indexed by the $1-\alpha$ quantile (roughly) of $\{E_i\}_{i \in \mathcal{D}^{\mathrm{cal}}}$. 
The procedure is outlined in Algorithm~\ref{alg:sc}. Note that the only computationally expensive component of this method is the estimation of the conditional histograms (see Supplementary Section~\ref{sec:supp-histograms} for details); the construction of the nested prediction intervals and the evaluation of the conformity scores have negligible cost because the optimization problem in~\eqref{eq:def-opt-problem} is an easy one.

It may be helpful to point out that, if $\pi$ provides an accurate representation of the true conditional distribution of $Y \mid X$, then the above conformity scores are uniformly distributed~\cite{romano2020classification}. In that ideal case, no calibration is needed and indeed our method simply reduces to applying~\eqref{eq:def-opt-problem} with $\tau =0.9$ to construct prediction intervals with $90\%$ coverage. In practice, however, $\pi$ can only be a possibly inaccurate estimate of $P_{Y \mid X}$ (hence why we will refer to it as $\hat{\pi}$ from now on), which means that the distribution of the conformity scores may not be uniform and the conformal calibration is necessary to obtain valid coverage.

The next result states that the output of our method has valid marginal coverage, regardless of the accuracy of $\hat{\pi}$. The proof relies on the sequence $S_t$ being nested; from there, coverage follows from the results of~\cite{romano2020classification,gupta2019nested}; see Supplementary Section~\ref{sec:supp-histograms}.

\begin{algorithm}[!htb]
  \SetAlgoLined
    \textbf{Input:} data $\left\{(X_i, Y_i)\right\}_{i=1}^{n}$, $X_{n+1}$, partition $\mathcal{B}$ of the domain of $Y$ into $m$ equal-sized bins, level $\alpha \in (0,1)$, resolution $T$ for the conformity scores, starting index $\bar{t}$ for recursive definition of conformity scores, black-box algorithm for estimating conditional distributions. \\
  Randomly split the training data into two subsets, $\mathcal{D}^{\mathrm{train}}, \mathcal{D}^{\mathrm{cal}}$. \\
    Sample $\varepsilon_i \sim \text{Uniform}(0,1)$ for all $i \in \{1,\ldots,n+1\}$, independently of everything else. \\
  Using the data in $\mathcal{D}^{\mathrm{train}}$, train any estimate $\hat{\pi}$ of the mass of $Y \mid X$ for each bin in $\mathcal{B}$~\eqref{eq:def-pi}; see Supplementary Section~\ref{sec:supp-histograms} for a concrete approach based on quantile regression. \\
   Compute $E_i  = E(X_i, Y_i, \varepsilon_i, \hat{\pi})$ for each $i \in \mathcal{D}^{\mathrm{cal}}$, with the function $E$ defined in~\eqref{eq:conformity-function}.\\
  Compute $\hat{t} = \hat{Q}_{1-\alpha}(\{E_i \}_{i \in \mathcal{D}^{\mathrm{cal}}})$ as the $\ceil{(1-\alpha)(1+|\mathcal{D}^{\mathrm{cal}}|)}$th smallest value in $\{E_i \}_{i \in \mathcal{D}^{\mathrm{cal}}}$. \\
  Select the $\hat{t}$-th element from $\{S_t(X_{n+1}, \varepsilon_{n+1}, \hat{\pi})\}_{t=0}^{T}$, defined in~\eqref{eq:def-s-start}--\eqref{eq:def-s-dec}:
  \begin{align*}
    \hat{C}^{\mathrm{sc}}_{n,\alpha}(X_{n+1}) = S_{\hat{t}} (X_{n+1}, \varepsilon_{n+1}, \hat{\pi}).
  \end{align*}\\
  \textbf{Output:} A prediction interval $\hat{C}^{\text{sc}}_{n,\alpha}(X_{n+1})$ for $Y_{n+1}$.
 \caption{CHR with split-conformal calibration}
 \label{alg:sc}
\end{algorithm}

\begin{theorem}[Marginal coverage] \label{thm:sc}
If $(X_i,Y_i)$, for $i \in \{1,\dots,n+1\}$, are exchangeable, then the output of~Algorithm~\ref{alg:sc} satisfies:
\begin{align}
  \mathbb{P}\left[Y_{n+1} \in \hat{C}^{\mathrm{sc}}_{n,\alpha}(X_{n+1}) \right] \geq 1-\alpha.
\end{align}
\end{theorem}
Note that Theorem~\ref{thm:sc} provides only a lower bound; a nearly matching upper bound on the marginal coverage can be generally established for split-conformal inference if the conformity scores are almost-surely distinct~\cite{vovk2005algorithmic,lei2018distribution,romano2019conformalized}. Although the CHR scores~\eqref{eq:conformity-function} are discrete, our experiments will show the coverage is tight as long as the resolution $T$ is not too small.

\section{Asymptotic analysis} \label{sec:asympotic}

We prove here that the prediction intervals computed by CHR (Algorithm~\ref{alg:sc}) are asymptotically equivalent, as $n \to \infty$, to those of the oracle from~\eqref{eq:oracle-interval}--\eqref{eq:oracle-interval-int}, if the model $\hat{\pi}$ is consistent for $P_{Y \mid X}$ and a few other technical conditions are met. In particular, we analyze a slightly modified version of Algorithm~\ref{alg:sc} in which there is no randomization; this is theoretically more amenable and equivalent in spirit, although it may yield wider intervals in finite samples.
Our theory relies on the additional Assumptions~\ref{assumption:iid}--\ref{assumption:smoothing}, explained below and stated formally in Supplementary Section~\ref{sec:supp-proofs}.
\begin{enumerate}
\item The samples are i.i.d., which is stronger than exchangeability; this is the key to our concentration results.
\item The black-box model estimates $P_{Y \mid X}$ consistently, in a sense analogous to that in~\cite{lei2018distribution,sesia2020comparison}. This assumption is crucial and may be practically difficult to validate in practice, but it can be justified by existing consistency results available for some models under suitable conditions, such as random forests~\cite{meinshausen2006quantile}. 
Further, the resolution $m$ of the partition of the $Y$ domain should grow with $n$ at a certain rate, and the resolution $T$ of the scores $E_i$ in~\eqref{eq:conformity-function} should grow as $T_n = n$.
\item The true $P_{Y \mid X}$ is continuous and with bounded density within a finite domain. This assumption is technical and could be relaxed with more work.
\item The true $P_{Y \mid X}$ is unimodal. This assumption is also technical and could be relaxed with more work.
\item The estimated conditional histogram $\hat{\pi}$ preserves the boundedness and unimodality of $P_{Y \mid X}$; this assumption may be unnecessary but it is convenient and quite innocuous at this point given Assumptions~\ref{assumption:consistency}--\ref{assumption:unimodality}.
\end{enumerate}
For simplicity, we assume the number of observations is $2n$, the test point is $(X_{2n+1}, Y_{2n+1})$, and $\mathcal{D}^{\mathrm{train}} = \mathcal{D}^{\mathrm{cal}} = n$, although different relative sample sizes would yield the same results. 

\begin{theorem}[Asymptotic conditional coverage and optimality] \label{thm:oracle-approximation}
$\forall \alpha \in (0,1]$, let $\hat{C}^{\mathrm{sc}}_{n,\alpha}(X_{2n+1})$ denote the prediction interval for $Y_{2n+1}$ computed by~Algorithm~\ref{alg:sc} at level $1-\alpha$ without randomization.
Under Assumptions~\ref{assumption:iid}--\ref{assumption:smoothing}, $\hat{C}^{\mathrm{sc}}_{n,\alpha}(X_{2n+1})$ is asymptotically equivalent, as $n \to \infty$, to $C^{\mathrm{oracle}}_{\alpha}(X_{2n+1})$, the output of the oracle~\eqref{eq:oracle-interval}--\eqref{eq:oracle-interval-int}. In particular, the following two properties hold.
\begin{enumerate} [label=(\roman*)]
  \item Asymptotic oracle length. For some sequences $\gamma_n \to 0$ and $\xi_n \to 0$ as $n \to \infty$,
\begin{align*}
  \P{|\hat{C}^{\mathrm{sc}}_{n,\alpha}(X_{2n+1})| \leq |C^{\mathrm{oracle}}_{\alpha}(X_{2n+1})| + \gamma_n} \geq 1 - \xi_n.
\end{align*}

  \item Asymptotic conditional coverage. For some sequences $\epsilon_n \to 0$ and $\zeta_n \to 0$ as $n \to \infty$,
\begin{align*}
  \P{ \P{Y \in \hat{C}^{\mathrm{sc}}_{n,\alpha}(X_{2n+1}) \mid X_{2n+1}} \geq 1-\alpha - \epsilon_n } \geq 1 - \zeta_n.
\end{align*}
\end{enumerate}
\end{theorem}

Theorem~\ref{thm:oracle-approximation} is similar to results in~\cite{lei2018distribution} and~\cite{sesia2020comparison} about the efficiency of earlier approaches to conformal regression, including CQR~\cite{sesia2020comparison}. However, the increased flexibility of our method is reflected by the oracle in~Theorem~\ref{thm:oracle-approximation}, which is stronger than those in \cite{lei2018distribution, sesia2020comparison}. In fact, the  oracle in~\cite{lei2018distribution} does not have conditional coverage, and that in~\cite{sesia2020comparison} produces wider prediction intervals with constant lower and upper miscoverage rates. Other conformal methods based on non-parametric density estimation~\cite{izbicki2019flexible,chernozhukov2019distributional} are not as efficient as CHR, in the sense that Theorem~\ref{thm:oracle-approximation} does not hold for them.

\section{Numerical experiments} \label{sec:experiments}

\subsection{Software implementation} \label{sec:software}

A Python implementation of CHR is available online 
at~\url{https://github.com/msesia/chr},
 along with code to reproduce the following numerical experiments. This software divides the domain of $Y$ into a desired number of bins with equal sizes, depending on the range of values observed in the training data; we use 100 bins for the synthetic data and 1000 for the real data. Then, we estimate the conditional histograms $\hat{\pi}$ using different black-box quantile regression models~\cite{taylor2000quantile,meinshausen2006quantile}, with a grid of quantiles ranging from 1\% to 99\%; see Supplementary Section~\ref{sec:supp-histograms}.
Our software also supports Bayesian additive regression trees~\cite{chipman2010bart} and could easily accommodate other alternatives.
For simplicity, we apply CHR and other benchmark methods by assigning equal numbers of samples to the training and calibration sets; this ensures all comparisons are fair, although different options may lead to even shorter intervals~\cite{sesia2020comparison}. See~\cite{bates2021testing} for a rigorous discussion of how this choice affects the variability of the coverage conditional on the calibration data, which is an issue we do not explore here.
See Supplementary Section~\ref{sec:supp-data} for details about how the models are trained, and information about the necessary computational resources. 

\subsection{Synthetic data} \label{sec:exp-synthetic}

We simulate a synthetic data set with a one-dimensional feature $X$ and a continuous response $Y$, from the same distribution previewed in Figure~\ref{fig:example-1}, which is similar to that utilized in~\cite{romano2019conformalized} to present CQR.
Our method is applied to 2000 independent observations from this distribution, using the first 1000 of them for training a deep quantile regression model, and the remaining ones for calibration.
Figure~\ref{fig:example-1} visualizes the resulting prediction bands for independent test data, comparing them to the analogous quantities output by CQR.
Both methods are based on the same neural network and guarantee 90\% marginal coverage, but ours leads to narrower intervals.
Indeed, the advantage of CHR is that it can extract information from all conditional quantiles estimated by the base model and then automatically adapt to the estimated data distribution. By contrast, CQR~\cite{romano2019conformalized} can only leverage a pre-specified lower and upper quantile (e.g., 5\% and 95\% in this example), and is therefore not adaptive to skewness.

Figure~\ref{fig:research-example-1-boxes}~(a) summarizes the performance of CHR over 100 experiments based on independent data sets, as a function of the sample size.
We evaluate the marginal coverage, approximate the worst-slab conditional coverage~\cite{cauchois2020knowing} as in~\cite{romano2020classification}, and compute the average interval width. 
We consider two benchmarks in addition to CQR~\cite{romano2019conformalized}: distributional conformal prediction (DCP)~\cite{chernozhukov2019distributional} and DistSplit~\cite{izbicki2019flexible}.
To facilitate the comparisons, all methods have the same base model.
 (We also applied DistSplit as implemented by~\cite{izbicki2019flexible}, with a different base model, but the version presented here performs better.)
These results show CHR leads to the shortest prediction intervals, while simultaneously achieving the highest conditional coverage. 
Compatibly with Theorem~\ref{thm:oracle-approximation}, the output of CHR becomes roughly equivalent to that of the omniscient oracle as the sample size grows; the latter can be implemented exactly here because we know the true data generating process.

Supplementary Figure~\ref{fig:research-example-1-boxes-naive} compares the performance of CHR in these experiments to that of naive uncalibrated 90\% prediction intervals based on the same deep neural network regression model, and obtained by simply plugging $\hat{\pi}$ into the oracle in~\eqref{eq:def-opt-problem}, with $\tau = 0.9$. Unsurprisingly, the naive prediction intervals do not generally have the desired marginal coverage; in this case, they tend to be too narrow if the sample size is small and too wide if the sample size is large. Although the lack of coverage for the uncalibrated intervals is not very pronounced here because this black-box model is relatively accurate even with $n=100$, such naive approach can yield arbitrarily low coverage in general, especially if the learning task is more difficult (e.g., for high-dimensional $X$), and is thus not reliable.

\begin{figure}[!htb]
  \centering
  \begin{minipage}[t]{0.03\textwidth}
    (a)
  \end{minipage}
  \begin{minipage}[]{0.95\textwidth}
    \includegraphics[width=\textwidth]{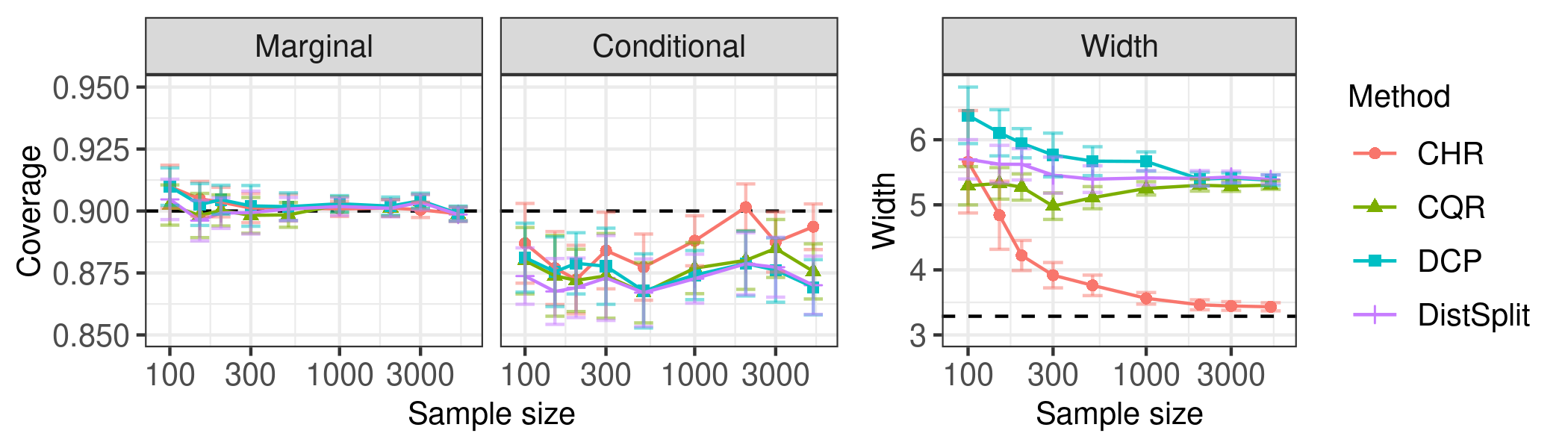}
  \end{minipage}

  \begin{minipage}[t]{0.03\textwidth}
    (b) 
  \end{minipage}
  \begin{minipage}[]{0.95\textwidth}
  \includegraphics[width=\textwidth]{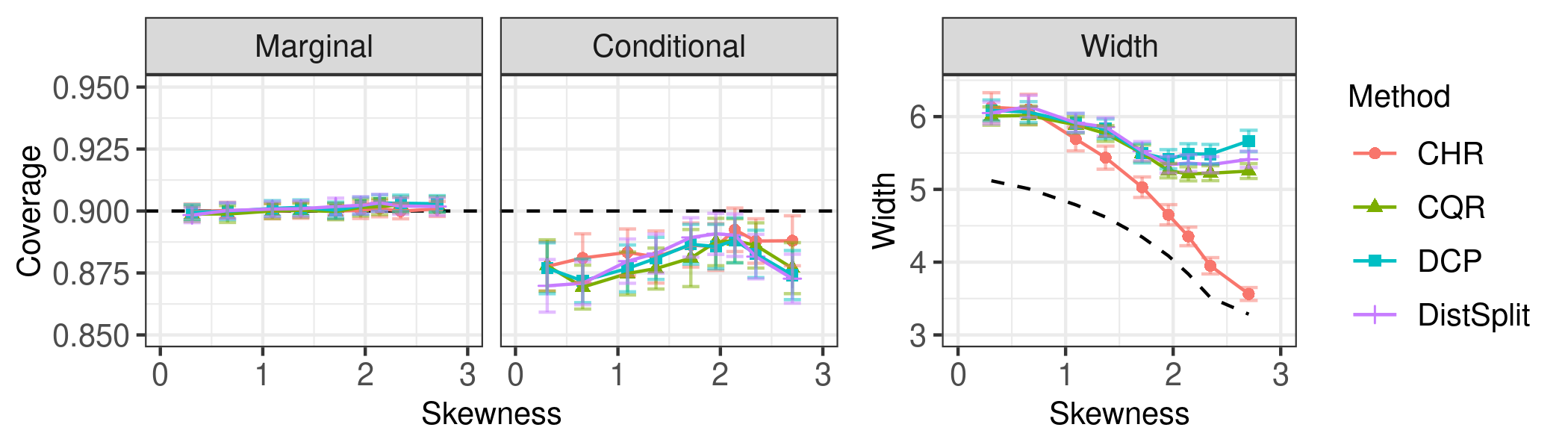}
  \end{minipage}
\caption{Performance of our method (CHR) and benchmarks on synthetic data distributed as in Figure~\ref{fig:example-1}. The dashed lines and curves correspond to an omniscient oracle. The vertical error bars span two standard errors from the mean. (a) Performance vs.~sample size. (b) Performance vs.~average skewness of the conditional distribution of the response, with a sample size of $1000$. The maximum skewness (near 3) matches that of the data in~(a).} 
  \label{fig:research-example-1-boxes}
\end{figure}

Figure~\ref{fig:research-example-1-boxes}~(b) shows analogous results from experiments in which we fix the sample size to 1000 and vary instead the skewness of the data distribution. Precisely, we flip a biased coin for each data point and transform $Y$ into $-Y$ if it lands heads, varying the coin bias as a control parameter. At one end of this spectrum, we recover the same skewed data distribution as in Figure~\ref{fig:research-example-1-boxes}~(a); at the other end, we have a symmetric $P_{Y \mid X}$. Our results are reported as a function of the expected skewness, defined as $\mathbb{E}[(Y-\mu(X))^3/\sigma^3(X) ]$, where $\mu(X)$ and $\sigma(X)$ are the mean and standard deviation of $Y \mid X$, respectively. These experiments show all methods are equivalent in terms of interval length if $P_{Y \mid X}$ is symmetric (skewness close to 0), while CHR can be much more powerful if $P_{Y \mid X}$ is skewed.

\subsection{Real data} \label{sec:real-data}

We apply CHR to the following seven public-domain data sets also considered in~\cite{romano2019conformalized}: physicochemical properties of protein tertiary structure (bio)~\cite{data-bio}, blog feedback (blog)~\cite{data-blog}, Facebook comment volume~\cite{data-facebook}, variants one (fb1) and two (fb2), from the UCI Machine Learning Repository~\cite{Dua:2019}; and medical expenditure panel survey number 19 (meps19)~\cite{data-meps19}, number 20 (meps20)~\cite{data-meps20}, and number 21 (meps21)~\cite{data-meps21}, from~\cite{cohen2009medical}. We refer to~\cite{romano2019conformalized} for more details about these data. As in the previous section, we would like to compare CHR to CQR, DistSplit, and DCP. 
However, as DCP~\cite{chernozhukov2019distributional} is unstable on all but one of these data sets, sometimes leading to very wide intervals, we replace it instead with a new hybrid benchmark that we call DCP-CQR.
This improves the stability of DCP by combining it with CQR~\cite{romano2019conformalized}, as explained in Supplementary Section~\ref{sec:supp-hybrid}.
This limitation of DCP may be explained by noting the method needs to learn a reasonably accurate approximation of the full conditional distribution of $Y \mid X$, and its performance is particularly sensitive to the estimation of the tails, which is most difficult; see Supplementary Section~\ref{sec:supp-hybrid} for more details.
 By contrast, CHR is robust because it only needs to estimate a histogram with relatively few bins---a much easier statistical task---and then it specifically focuses on finding the shortest intervals containing high probability mass.
We apply all methods, including our CHR, based on the same deep quantile regression model.
Their performances are evaluated as in the previous section, averaging over 100 independent experiments per data set. In each experiment, 2000 samples are used for training, 2000 for calibration, and the remaining ones  for training. All features are standardized to have zero mean and unit variance. The nominal coverage rate is 90\%.

\begin{figure}[!htb]
  \centering
    \includegraphics[width=\textwidth]{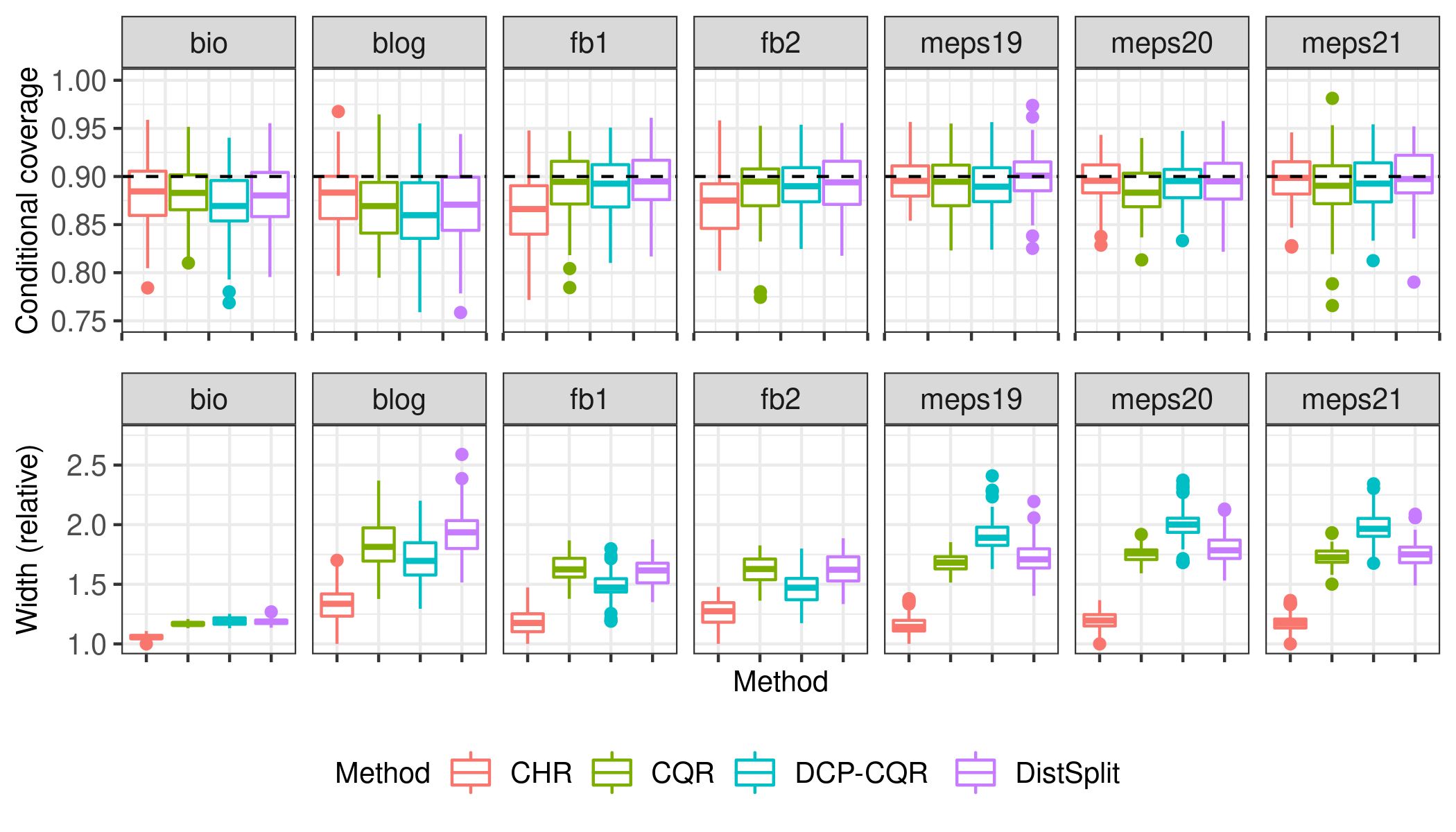}
    \caption{Performance of our method (CHR) and benchmarks on several real data sets, using a deep neural network model. All methods provably have 90\% marginal coverage. The box plots show the distribution of results over 100 random test sets, each containing 2000 observations. } 
  \label{fig:results_real_nnet}
\end{figure}

Figure~\ref{fig:results_real_nnet} shows the distribution of the conditional coverage and interval width corresponding to different methods, separately for each data set. To simplify the plots by using a shared vertical axis, the widths of the prediction intervals are scaled, separately for each data set, so that the smallest one is always equal to one.
Marginal coverage is omitted here because all methods provably control it; however, it can be found in Supplementary Table~\ref{tab:results_real}.
All methods perform well in terms of worst-slab conditional coverage. 
CHR outperforms the others in terms of statistical efficiency, as it consistently leads to the shortest intervals. CQR and DistSplit are comparable to each other, while DCP-CQR sometimes outputs wider intervals. Supplementary Figure~\ref{fig:results_real_rf} shows that analogous results are obtained if we utilized a random forest model instead of a deep neural network. Supplementary Table~\ref{tab:results_real} summarizes these results in more detail, including marginal coverage and the omitted performance of the original DCP.
Finally, Supplementary Figure~\ref{fig:results_real_nnet_naive} compares the performance of CHR in these experiments with real data to that of naive uncalibrated 90\% prediction intervals based on the same deep neural network regression model, as in Figure~\ref{fig:research-example-1-boxes-naive}. These results show that the naive prediction intervals do not generally have the desired marginal coverage; in some cases they are too narrow, and in others they are too wide.

\section{Conclusions} \label{sec:conclusions}


This paper developed CHR, a non-parametric regression method based on novel conformity scores leading to shorter prediction intervals with coverage, and enjoying stronger asymptotic efficiency, compared to the state-of-the-art alternatives. 
Of course, real data are finite and our theory relies on assumptions which may be difficult to validate; nonetheless, it is a sanity check and it provides an informative comparison. Further, the experiments confirm CHR performs well in practice.

The ability of CHR to automatically adapt to unknown skewness may prove useful in practice, as empirical data often follow distributions with power-law tails~\cite{clauset2009power}. Indeed, the data sets analyzed in Section~\ref{sec:real-data} tend to have highly skewed outcomes with many observations equal to zero.
At the same time, a limitation of CHR is that it does not rigorously control the lower and upper miscoverage rates, which may be important for some applications; if that is the case, the modified version of CQR proposed by~\cite{romano2019conformalized} would be a better choice. Note that the standard implementations of CQR and of the other benchmarks~\cite{chernozhukov2019distributional,izbicki2019flexible} considered in this paper are not guaranteed to separately control the lower and upper miscoverage rates. In any case, users of our method could naturally obtain approximations of the lower and upper miscoverage rates for any prediction interval by looking at the underlying conditional histograms, although these estimates are of course not calibrated in finite samples.

Algorithm~\ref{alg:cv+} in Supplementary Section~\ref{sec:supp-cv+} extends CHR to accommodate cross-validation+~\cite{barber2019predictive}, which is often more powerful, and computationally expensive, compared to the split-conformal approach presented in this paper. The strategy is the same as that followed by~\cite{romano2020classification} in the classification setting, although it requires an extra step, in which the standard cross-validation+ prediction set~\cite{barber2019predictive} is replaced by its convex hull to guarantee the final output is an interval~\cite{gupta2019nested}.
Supplementary Theorem~\ref{thm:CV+} establishes that Algorithm~\ref{alg:cv+} leads to marginal coverage above $1-2\alpha$, applying the more general theory from~\cite{gupta2019nested}.
Finally, a possible directions for future research may involve the extension of our method to deal with multi-dimensional responses $Y$.

From a broader perspective, this paper falls within a rapidly growing body of works focusing on improving the interpretability and statistical reliability of machine learning algorithms. Prediction intervals with marginal coverage provide a convenient way of communicating uncertainty about the accuracy of any machine learning black-box, which is important to increase their reliability, to ensure their fairness~\cite{romano2020malice}, and to facilitate their acceptance. Furthermore, conformal prediction intervals provide a principled metric by which to compare different machine learning algorithms~\cite{holland2020making}.


\begin{ack}
The authors are grateful to Stephen Bates, Emmanuel Cand{\`e}s, and Wenguang Sun for providing insightful comments about an earlier version of this manuscript.
M.S.~thanks the center for Advanced Research Computing at the University of Southern California for providing computing resources. Y.R.~was supported by the Israel Science Foundation (grant 729/21). Y.R.~also thanks the Career Advancement Fellowship, Technion, for providing research support.
\end{ack}

\bibliographystyle{abbrv}
\bibliography{biblio}

\pagebreak

\setcounter{equation}{0}
\setcounter{figure}{0}
\setcounter{table}{0}
\setcounter{page}{1}
\setcounter{section}{0}
\renewcommand{\thesection}{S\arabic{section}}
\renewcommand{\theequation}{S\arabic{equation}}
\renewcommand{\thetheorem}{S\arabic{theorem}}
\renewcommand{\thelemma}{S\arabic{lemma}}
\renewcommand{\thetable}{S\arabic{table}}
\renewcommand{\thefigure}{S\arabic{figure}}
\renewcommand{\thealgocf}{S\arabic{algocf}}

\begin{center}
  \textbf{\Large Supplementary Material for:\\Conformal Prediction using Conditional Histograms}\\[.2cm]
  Matteo Sesia$^{1}$, and Yaniv Romano$^{2}$\\[.1cm]
  {\itshape ${}^1$Department of Data Sciences and Operations, University of Southern California, USA\\
  ${}^2$Departments of Electrical and Computer Engineering and of Computer Science, Technion, Israel\\}
\end{center}

\section{Supplementary methods} \label{sec:supp-methods}

\subsection{Estimating conditional distributions and histograms} \label{sec:supp-histograms}

For any fixed $K>1$, define the sequence $a_k = k/K$ for $k \in \{0,\ldots,K\}$, and let $\hat{q}(x) = (\hat{q}_{a_0}(x), \ldots, \hat{q}_{a_K}(x))$ denote a collection of $K+1$ conditional quantile estimators,\footnote{Recall the definition of conditional quantiles: each $\hat{q}_{c}(x)$ is an estimate of the true $c$-th conditional quantile of $Y \mid X=x$: that is, the smallest value of $y$ such that $\P{Y \leq y \mid X=x} \geq c$.} where $\hat{q}_{a_k}(x)$ attempts to approximate the $a_k$-th quantile of the conditional distribution of $Y \mid X=x$, such that $\hat{q}_{a_{k}}(x) \leq \hat{q}_{a_{k+1}}(x)$ for all $k$ and $x$. Note that we allow multiple estimated quantiles to be identical to each other, to accommodate the possibility of point masses.
Furthermore, we assume $\hat{q}_0(x)$ and $\hat{q}_1(x)$ are conservative upper and lower bounds for the support of $Y\mid X=x$, i.e., $\hat{q}_0(X) = b_0 < Y < b_m = \hat{q}_1(X)$.
We will discuss in the next section practical options for estimating $\hat{q}(x)$.
Now, we leverage any given $\hat{q}(x)$ to compute estimates $\hat{\pi}_j(x)$ of the unknown bin probabilities $\pi_j(x)$ in~\eqref{eq:def-pi}, for all $j \in \{1,\ldots,m\}$.
Although there are multiple way of doing this, a principled solution is to convert the information contained in $\hat{q}$ into a piece-wise constant density estimate, and then integrate that density within each bin.
Precisely, for any fixed $x$, let $\hat{c}(x) = (\hat{c}_0(x), \ldots, \hat{c}_{\bar{m}(x)}(x))$ denote the strictly increasing sequence of $\bar{m}(x) \leq m$ unique values in $\hat{q}(x)$, and define our estimated conditional density $\hat{f}$ as
\begin{align*}
  \hat{f}(y \mid x) = \frac{1}{\bar{m}(x)} \sum_{j=1}^{\bar{m}(x)} h_j(x) \I{\hat{c}_{j-1}(x) < y < \hat{c}_{j}(x)},
\end{align*}
with
\begin{align*}
  h_{j}(x) = \frac{\# \{ j' \in \{0,\ldots,m\} : \hat{q}_{a_{j'}}(x) = \hat{c}_j(x)\}}{m \cdot \left[ \hat{c}_{j}(x) - \hat{c}_{j-1}(x) \right]}.
\end{align*}
Intuitively, $\hat{f}$ is a histogram with $\bar{m}(x)$ bins, whose delimiters are $(\hat{c}_0(x), \ldots, \hat{c}_{\bar{m}(x)}(x))$ and whose heights are $(h_1(x), \ldots, h_{\bar{m}(x)}(x))$. The numerator in the expression for $h_{j}(x)$ counts the number of estimated quantiles that are identical to the $\hat{c}_{j}$-th one, accounting for the possible presence of point masses in the approximation of $P_{Y \mid X}$ captured by $\hat{q}(x)$.

As the tails of the above estimated conditional density may be particularly inaccurate because relatively little information is available to estimate extremely low or high quantiles, we smooth them. This ensures any estimation errors will not make $\hat{f}$ decay too fast, forcing one to look much farther than necessary in the tails before finding sufficient mass for the desired prediction intervals. The smoothing approach we adopt simply consists of making $\hat{f}$ constant between $b_0$ (the uniform lower bound on $Y$) and the 1\% quantile, as well as between the 99\% quantile and $b_m$ (the uniform upper bound on $Y$), distributing these 1\% probability masses uniformly in the tails.\footnote{We thank Stephen Bates for suggesting a smoothing strategy which inspired this solution.}

We utilize the same estimated conditional distribution thus obtained for our method as well as for our implementations of DCP and DistSplit, because it performs relatively well for all of them. In particular, our method leverages this distribution to construct a conditional histogram as follows.
The probability mass $\hat{\pi}_j(x)$ for the bin $[b_{j-1},b_j)$ is given by:
\begin{align} \label{eq:histogram-estimator}
  \hat{\pi}_j(x) = \int_{b_{j-1}}^{b_j} \hat{f}(y \mid x) dy,
\end{align}
which is easy to compute because $\hat{f}$ is piece-wise constant.
Finally, $\sum_{j=1}^{m} \hat{\pi}_j(x) = 1$ by construction.

When implemented with a deep neural network~\cite{moon2021learning}, the multi-quantile regression method described above has computational cost comparable to that of the bi-quantile regression model utilized by CQR~\cite{romano2019conformalized}. 
Indeed, the numbers of parameters and the architecture of the neural network are essentially the same in
both cases, the only difference is that our model has a wider output layer.
Therefore, the computational cost and training runtime are
approximately the same. Intuitively, this can be understood as thinking of the neural network as
learning an approximate representation of the conditional distribution of $Y \mid X$, regardless of how
many different quantiles are explicitly estimated. Of course, that is not to say that estimating many
quantiles is as easy as estimating only two, but most of the additional statistical difficulty would
come from estimating extremely large or small quantiles, not the intermediate ones. Precisely to
avoid this problem, our model does not attempt to estimate extremely large or small quantiles
(below 1\% or above 99\%); instead, the tails are smoothed as explained above.

\subsection{Randomized prediction intervals} \label{sec:supp-random}

Due to the discrete nature of the optimization problem in~\eqref{eq:def-opt-problem}, the inequality involving $\tau$ may not be binding at the optimal solution. 
Therefore, to avoid producing wider intervals than necessary, we introduce some randomization. Let $\varepsilon$ be a uniform random variable between 0 and 1 drawn independently of everything else. Then, define the following function $R$, which takes as input $[l,u] \subseteq \{1,\ldots,m\}$, $x$, $\pi$, $\varepsilon$, $\tau$, and outputs a sub-interval of $\{1,\ldots,m\}$:
\begin{align} \label{eq:def-randomized}
  R( [l,u], x, \varepsilon, \pi, \tau) \defeq
  \begin{cases}
    [l,u], & \text{if } \varepsilon > V([l,u], x, \pi, \tau), \\
    [l-1, u] , & \text{if } \varepsilon \leq V([l,u], x, \pi, \tau) \text{ and } \pi_l(x) \leq \pi_u(x), \\
    [l, u-1], & \text{if } \varepsilon \leq V([l,u], x, \pi, \tau) \text{ and } \pi_l(x) > \pi_u(x),
  \end{cases}
\end{align}
where the function $V$ is given by
\begin{align*}
  V([l,u], x, \pi, \tau) \defeq \frac{\sum_{j=l}^{u} \pi_j(x) - \tau}{\min\left\{ \pi_l(x), \pi_u(x) \right\}}.
\end{align*}
In words, $R$ returns a random subset of $[l,u]$ by removing the extreme bin with the smallest mass according to $\pi$, based on the outcome of a biased coin flip, if the total mass in the original interval exceeds $\tau$. Consequently, the total mass in $R( [l,u], x, \varepsilon, \pi, \tau)$ will on average be exactly equal to $\tau$ if $[l,u]$ is given by~\eqref{eq:def-opt-problem}. 

Inspired by the above oracle, the randomized version of our algorithm is implemented as follows.
First, fix any {\em starting point} $\bar{t} \in \{0,\ldots,T\}$ and define $S_{\bar{t}}$ by applying \eqref{eq:def-opt-problem} and \eqref{eq:def-randomized} without the nesting constraints (with $S^- = \emptyset$ and $S^+ = \{1,\ldots,m\}$):
\begin{align} \label{eq:def-s-start}
  & S^0_{\bar{t}}  \defeq \mathcal{S}(x, \pi, \emptyset, \{1,\ldots,m\}, \tau_{\bar{t}}),
  & S_{\bar{t}}  \defeq R( S^0_{\bar{t}}, x, \varepsilon, \pi, \tau_{\bar{t}}).
\end{align}
Having computed the initial interval $S_{t}$ for $t = \bar{t}$, we recursively extend the definition to the wider intervals indexed by $t = \bar{t} + 1, \ldots, T$ as follows:
\begin{align} \label{eq:def-s-inc}
\begin{split}
  S^0_{t} & \defeq \mathcal{S}(x, \pi, S_{t-1}, \{1,\ldots,m\}, \tau_{t}), \\
  S_{t} & \defeq 
                    \begin{cases}
                      R( S^0_{t}, x, \varepsilon, \pi, \tau_{t}), & \text{if } S_{t-1} \subseteq R( S^0_{t}, x, \varepsilon, \pi, \tau_{t}), \\
                      S^0_{t}, & \text{otherwise.}
                    \end{cases}
                  \end{split}
\end{align}
Intuitively, the randomization step in~\eqref{eq:def-s-inc} is applied only if it does not violate the nesting constraints, ensuring $S_{\bar{t}} \subseteq S_{\bar{t}+1} \subseteq \ldots \subseteq S_{T}$. See the top row of Figure~\ref{fig:intervals-schematic} for a schematic of this step.
Similarly, the narrower intervals $S_{t}$ indexed by $t = \bar{t}-1, \bar{t}-2, \ldots 0$ are defined recursively as:
\begin{align} \label{eq:def-s-dec}
  & S^0_{t} \defeq \mathcal{S}(x, \pi, \emptyset, S^{0}_{t+1}, \tau_{t}),
  & S_{t} \defeq R( S^0_{t}, x, \varepsilon, \pi, \tau_{t}).
\end{align}
See the bottom row of Figure~\ref{fig:intervals-schematic} for a schematic of this step.
Note that $\mathcal{S}$ in~\eqref{eq:def-opt-problem} is applied here in~\eqref{eq:def-s-dec} with $S^+ = S^{0}_{t+1}$ to ensure the optimization problem has a feasible solution; this may not necessarily be the case with $S^+ = S_{t+1}$, as the latter is randomized and may therefore sometimes contain less mass than necessary, according to the input $\pi$. Nonetheless, the sequence of intervals $\{ S_{t} \}_{t=0}^{T}$ thus obtained is provably nested, as previewed in Figure~\ref{fig:intervals-schematic}.
In the following, it will be convenient to highlight the dependence of this sequence on $x,\varepsilon,\pi$ by writing it as $S_t(x,\varepsilon,\pi)$.

\begin{proposition} \label{prop:nested-sequence}
The sequence of intervals $\{ S_{t} \}_{t=0}^{T}$ defined recursively by~\eqref{eq:def-s-start}--\eqref{eq:def-s-dec}, and depending on $x, \varepsilon, \pi$, always satisfies $S_{t-1} \subseteq S_t$ for all $t \in \{1,\ldots,T\}$.
\end{proposition}

Proposition~\ref{prop:nested-sequence} is proved below. Again, note that this results holds regardless of the starting point $\bar{t}$ in~\eqref{eq:def-s-start}, although the most intuitive choice is to pick $\bar{t}$ such that $\tau_{\bar{t}} \approx 1-\alpha$.

\begin{proof}[Proof of Proposition~\ref{prop:nested-sequence}]
First, we show $S_{t-1} \subseteq S_t$ for all $t = \bar{t} + 1, \ldots, T$.
We know from~\eqref{eq:def-s-inc} that there are two possibilities. (i) If $S_{t-1} \subseteq R( S^0_{t}, x, \varepsilon, \pi, \tau_{t})$, then $S_{t} = R( S^0_{t}, x, \varepsilon, \pi, \tau_{t})$ and so $S_{t-1} \subseteq S_{t}$. (ii) Otherwise, $S_{t} = S^0_{t} = \mathcal{S}(x, \pi, S_{t-1}, \{1,\ldots,m\}, \tau_{t})$, which contains $S_{t-1}$ by definition of $\mathcal{S}$ in~\eqref{eq:def-opt-problem}.

Second, we show $S_{t} \subseteq S_{t+1}$ for all $t = \bar{t} - 1, \ldots, 0$, using~\eqref{eq:def-s-dec}. Here, we can also distinguish between two possibilities. (i) If $S_{t+1} = S^{0}_{t+1}$, then $S^0_{t} \subseteq S_{t+1}$ by definition of $\mathcal{S}$ in~\eqref{eq:def-opt-problem}, and so $S_{t} \subseteq S_{t+1}$ because $S_{t} \subseteq S^0_{t}$. (ii) Otherwise, $S_{t+1}$ must have been randomized and this is the least obvious case on which we focus below. 

Suppose $S_{t+1} \subset S^{0}_{t+1}$. We know from the definition of $\mathcal{S}$ in~\eqref{eq:def-opt-problem} that $S^0_{t} \subseteq S^0_{t+1}$. On the one hand, if $S^0_{t} = S^0_{t+1}$, it is easy to see from~\eqref{eq:def-randomized} that $S_{t} = S_{t+1}$ because $\tau_{t} < \tau_{t+1}$, and so the same bin randomly removed from $S^0_{t+1}$ will also certainly be removed from $S_{t}$. On the other hand, if $S^0_{t} \subset S^0_{t+1}$, it must be the case that $S^0_{t} \subseteq S_{t+1}$ because $S_{t+1}$ is obtained by removing the boundary bin of $S^0_{t+1}$ with the smallest mass. Therefore, $S^0_{t}$ cannot include the aforementioned bin without also satisfying $S^0_{t} = S^0_{t+1}$, for otherwise it would be possible to find an alternative $S^{0'}_{t}$ with equal length and smaller but still feasible mass above $\tau_{t}$, which is inconsistent with optimality of $S^0_{t}$ according to the definition of $\mathcal{S}$ in~\eqref{eq:def-opt-problem}. This implies $S_{t} \subset S^0_{t+1}$ because $S_{t} \subset S^0_{t}$, completing the proof.
\end{proof}

\subsection{The DCP-CQR hybrid method} \label{sec:supp-hybrid}

The DCP-CQR hybrid repurposes the DCP calibration algorithm~\cite{chernozhukov2019distributional} to adaptively choose which lower and upper estimated quantiles should be extracted from the machine-learning model; then, it takes these as a starting point for CQR~\cite{romano2019conformalized}.
By contrast, the original CQR requires one to pre-specify which two conditional quantiles should be estimated by the machine learning model. For example, we implement CQR by estimating the $\alpha/2$ and $1-\alpha/2$ quantiles, as this is the most intuitive choice and it guarantees the method is asymptotically efficient~\cite{sesia2020comparison}, although in a weaker sense compared to the oracle property established by Theorem~\ref{thm:oracle-approximation} for CHR.

The reason why DCP-CQR is more stable than DCP is that our hybrid only considers a limited grid of possible quantiles (e.g., 1\% to 99\%). If the machine learning model is very inaccurate and the fixed quantile grid turns out to be insufficient to reach 90\% coverage (assuming $\alpha = 0.1$) on the calibration data, then we can simply rely on CQR to correct the coverage by adding a constant shift to the prediction bands. By contrast, the original DCP~\cite{chernozhukov2019distributional} may sometimes rely on extreme quantiles (e.g., 99.99\%) of the conditional distribution estimated by the fitted model, which are unreliable. 

\subsection{Calibration with cross-validation+} \label{sec:supp-cv+}

Algorithm~\ref{alg:cv+} extends Algorithm~\ref{alg:sc} to accommodate a calibration scheme alternative to data splitting: cross-validation+~\cite{barber2019predictive}. While we do not fully review cross-validation+ for lack of space, readers aware of the work of~\cite{barber2019predictive}, or~\cite{gupta2019nested}, will recognize this as a straightforward combination of their techniques with our novel conformity scores.

\begin{algorithm}[!htb]
  \SetAlgoLined
  \textbf{Input:} data $\left\{(X_i, Y_i)\right\}_{i=1}^{n}$, $X_{n+1}$, partition $\mathcal{B}$ of the domain of $Y$, level $\alpha \in (0,1)$, resolution $T$ for the conformity scores, starting index $\bar{t}$ for recursive definition of conformity scores, machine-learning algorithm for estimating conditional distributions. \\
  Randomly split the training data into $K$ disjoint subsets, $\mathcal{D}_1,\dots,\mathcal{D}_K$, each of size $n/K$. \\
  Sample $\varepsilon_i \sim \text{Uniform}(0,1)$ for each $i \in \{1,\ldots,n+1\}$, independently of everything else. \\
  \For{$k \in \{1, \ldots, K\}$} {
  Train any estimate $\hat{\pi}^k$ of the mass of $Y \mid X$ for each bin in $\mathcal{B}$, e.g., with~\eqref{eq:histogram-estimator}, based on all data points except those in $\mathcal{D}_k$.
    }
  Use the function $E$ defined in~\eqref{eq:conformity-function} to construct the prediction interval
  \begin{align} \label{eq:predictions-CV+}
    \hat{C}^{\text{CV+}}_{n,\alpha}(X_{n+1}) = \text{Conv}\left( C  \right),
  \end{align}
  where  $\text{Conv}(C)$ is the convex hull of the set $C$, which is defined as
  \begin{align}
    C
    = \left\{ y: \frac{1}{n} \sum_{i=1}^n \mathbf{1}\left[ E(X_i,Y_i,\varepsilon_i, \hat{\pi}^{k(i)}) < E(X_{n+1},y,\varepsilon_{n+1}, \hat{\pi}^{k(i)}) \right] < 1 - \alpha_n \right\},
  \end{align}
  with $\alpha_n = \alpha (1+1/n) - 1/n $ and $k(i) \in \{1,\ldots,K\}$ is the fold containing the $i$-th sample. \\
  \textbf{Output:} A prediction interval $\hat{C}^{\text{CV+}}_{n,\alpha}(X_{n+1})$ for the unobserved label $Y_{n+1}$.
 \caption{CV+ adaptive predictive intervals for regression}
 \label{alg:cv+}
\end{algorithm}

\begin{theorem}[Adapted from Theorem~3 in~\cite{gupta2019nested}]  \label{thm:CV+}
Under the same assumptions of Theorem~\ref{thm:sc}, if $\hat{\pi}$ is invariant to permutations of its input samples, the output of~Algorithm~\ref{alg:cv+} satisfies:
\begin{align}
  \mathbb{P} \left[ \{Y_{n+1} \in \hat{C}^{\mathrm{CV+}}_{n,\alpha}(X_{n+1}) \right] \geq 1-2\alpha - \min \left\{\frac{2(1-1/K)}{n/K+1}, \frac{1-K/n}{K+1} \right\}.
\end{align}
In the special case where $K=n$, this bound simplifies to:
\begin{align}
  \mathbb{P} \left[ Y_{n+1} \in \hat{C}^{\mathrm{JK+}}_{n,\alpha}(X_{n+1}) \right] \geq 1-2\alpha.
\end{align}
\end{theorem}

\FloatBarrier

\section{Theoretical analysis} \label{sec:supp-proofs}

\subsection{Finite-sample analysis}

\begin{proof}[Proof of Theorem~\ref{thm:sc}]
  The interval $\hat{C}^{\mathrm{sc}}_{n,\alpha}(X_{n+1})$ is such that $Y_{n+1} \in \hat{C}^{\mathrm{sc}}_{n,\alpha}(X_{n+1})$ if and only if
  \begin{align*}
    \min \left\{ t \in \{0,\ldots,T\} : Y_{n+1} \in S_t(X_{n+1},\varepsilon_{n+1}, \hat{\pi}) \right\} \leq \hat{Q}_{1-\alpha}(\{E_i\}_{i \in \mathcal{D}^{\mathrm{cal}}}).
  \end{align*}
  Equivalently, $Y_{n+1} \in \hat{C}^{\mathrm{sc}}_{n,\alpha}(X_{n+1})$ if and only if 
  \begin{align} \label{eq:conformity-quantile}
    E_{n+1} \leq \hat{Q}_{1-\alpha}(\{E_i\}_{i \in \mathcal{D}^{\mathrm{cal}}}).
  \end{align}
  The proof is standard from here: the key idea is that the probability of the event in~\eqref{eq:conformity-quantile} is at least $1-\alpha$ because all conformity scores $\{E_i\}_{i=1}^{n+1}$ are exchangeable; see~\cite{romano2019conformalized} for details.
\end{proof}

\subsection{Asymptotic analysis}

\begin{assumption}[i.i.d.~data] \label{assumption:iid}
The data $\{(X_i,Y_i)\}_{i=1}^{2n+1}$ are i.i.d.~from some unknown joint distribution.
\end{assumption}

\begin{assumption}[consistency] \label{assumption:consistency}
For any fixed $n$, let $m_n$ denote the number of bins in the partition $\mathcal{B}$ of the space of $Y$ utilized by our method.
Let $F(y \mid x)$ denote the cumulative distribution function of $Y \mid X=x$, and define $\hat{F}(y \mid x)$ as the estimate of the latter according to $\hat{\pi}$ from~\eqref{eq:histogram-estimator}; i.e., 
\begin{align*}
  \hat{F}(y \mid x) \defeq \sum_{j=1}^{\hat{j}(y)} \hat{\pi}_j(x),
\end{align*}
where $\hat{j}(y) = \max \{ j \in \{1,\ldots,m_n\} : y \leq b_{j}\}$.
Then, assume there exists a sequence $\eta_n \to 0$, as $n \to \infty$, such that, for all $j \in \{1,\ldots,m_n\}$,
\begin{align} \label{eq:consistency}
    \P{\E{\left( \hat{F}(b_j \mid X) - F(b_j \mid X) \right)^2 \mid \mathcal{D}^{\mathrm{train}}} \leq \eta_{n}^{2} }
    & \geq 1-\eta_{n}^{2}.
\end{align}
Further, $m_n = \lfloor \eta_n^{-1} \rfloor$ and $T_n = n$, where $T_n$ is the resolution of the conformity scores $E_i$~\eqref{eq:conformity-function}.
\end{assumption}

\begin{assumption}[regularity] \label{assumption:regularity}
For any $x \in \mathbb{R}^p$, the conditional distribution of $Y \mid X=x$ is continuous with density $f(y \mid x)$ and support $[-C,C]$, for some finite $C>0$. Furthermore, $1/K < f(y \mid x) < K/2$ within $[-C,C]$, for some $K>0$.
\end{assumption}

\begin{assumption}[unimodality] \label{assumption:unimodality}
For any $x \in \mathbb{R}^p$, the conditional distribution of $Y \mid X=x$ is unimodal; i.e., there exists $y_0 \in [-C,C]$ (depending on $x$), such that $f(y_0 + y'' \mid x) \leq f(y_0 + y')$ if $y'' \geq y' \geq 0$, and $f(y_0 + y'' \mid x) \leq f(y_0 + y')$ if $y'' \leq y' \leq 0$.
\end{assumption}

\begin{assumption}[smoothing] \label{assumption:smoothing}
For any fixed $n$ and $x \in \mathbb{R}^p$, the estimated conditional distribution of $Y \mid X=x$ characterized by $\hat{\pi}(x)$ is unimodal. That it, there exists $j_0 \in \{1,\ldots,m_n\}$ such that $\hat{\pi}_{j_0 + k''}(x) \leq \hat{\pi}_{j_0 + k'}(x)$ if $k'' \geq k' \geq 0$, and $\hat{\pi}_{j_0 + k''}(x) \leq \hat{\pi}_{j_0 + k'}(x)$ if $y'' \leq y' \leq 0$, for all $k'',k'$ such that $j_0+k'' \in \{1,\ldots,m\}$ and $j_0+k' \in \{1,\ldots,m\}$. Furthermore, assume $\hat{\pi}_j \leq K$ for all $j \in \{1,\ldots,m_n\}$, for any $n$.
\end{assumption}

Note that, if $\hat{\pi}$ is based on a quantile model as described in Section~\ref{sec:supp-histograms}, Assumption~\ref{assumption:consistency} is closely related to the consistency assumption on the estimated conditional quantiles utilized by~\cite{sesia2020comparison} to study CQR~\cite{romano2019conformalized}, although the latter only involved two fixed quantiles. More precisely, leveraging Assumption~\ref{assumption:regularity}, one could rewrite~\eqref{eq:consistency} in terms of the consistency of the underlying quantile regressors,
  \begin{align} \label{eq:consistency-quantiles}
    \P{\E{\left( \hat{q}_{\tau_t}(X) - q_{\tau_t}(X) \right)^2 \mid \mathcal{D}^{\mathrm{train}}} \leq \tilde{\eta}_{n} }
    & \geq 1-\tilde{\rho}_{n},
  \end{align}
for some sequences $\tilde{\eta}_{n} \to 0$ and $\tilde{\rho}_{n} \to 0$ as $n \to \infty$.
The assumption in~\eqref{eq:consistency-quantiles} is also similar to that adopted in~\cite{lei2018distribution} for mean regression estimators, and it is weaker than requiring consistency in the sense of $L^2$ convergence, by Markov's inequality.

\subsubsection*{Main result}

\begin{theorem}[More precise restatement of Theorem~\ref{thm:oracle-approximation}] \label{thm:supp-oracle-approximation}
For any $\alpha \in (0,1]$, let $\hat{C}^{\mathrm{sc}}_{n,\alpha}(X_{2n+1})$ denote the prediction interval at level $1-\alpha$ for $Y_{2n+1}$ obtained by applying Algorithm~\ref{alg:sc} with $\varepsilon_i= 0$ for all $i \in \{n+1, \ldots,2n+1\}$; that is, we omit the randomization in~\eqref{eq:def-randomized}.
Under Assumptions~\ref{assumption:iid}--\ref{assumption:smoothing}, the prediction interval $\hat{C}^{\mathrm{sc}}_{n,\alpha}(X_{2n+1})$ is asymptotically equivalent, as $n \to \infty$, to $C^{\mathrm{oracle}}_{\alpha}(X_{2n+1})$---the output of the ideal oracle from~\eqref{eq:oracle-interval}--\eqref{eq:oracle-interval-int}. In particular, the following two properties hold.
\begin{enumerate} [label=(\roman*)]
  \item Asymptotic oracle length, in the sense that
\begin{align*}
  \P{|\hat{C}^{\mathrm{sc}}_{n,\alpha}(X_{2n+1})| \leq |C^{\mathrm{oracle}}_{\alpha}(X_{2n+1})| + \gamma_n} \geq 1 - \xi_n,
\end{align*}
where $\gamma_n = 4 C \eta_n + K \left(\epsilon_n + 2 \eta_n^{1/3} \right)  \to 0$, and $\xi_n = \delta_n + 2 n^{-2}  \to 0$.

\item Asymptotic oracle conditional coverage, in the sense that
\begin{align*}
  \P{ \P{Y \in \hat{C}^{\mathrm{sc}}_{n,\alpha}(X_{2n+1}) \mid X_{2n+1}} \geq 1-\alpha - \epsilon_n } \geq 1 - \zeta_n,
\end{align*}
where $\epsilon_n = 2/n + 5 \eta_n^{1/3} + (1 + 2K) \eta_n + 2 \sqrt{(\log n)/n}\to 0$ and $\zeta_n = \eta_n^{1/3} + \eta_n  + 2 n^{-2} \to 0$.
\end{enumerate}
\end{theorem}

\begin{proof}[\textbf{Proof of Theorem~\ref{thm:supp-oracle-approximation}}]

Assumption~\ref{assumption:unimodality} (unimodality) and Assumption~\ref{assumption:smoothing} (smoothness) imply the optimal intervals solving~\eqref{eq:def-opt-problem} for different values of $\tau$ are nested, so we do not need to define the prediction intervals recursively.
More precisely, under Assumptions~\ref{assumption:unimodality} and~\ref{assumption:smoothing}, $\hat{C}^{\mathrm{sc}}_{n,\alpha}(X_{2n+1}) = \hat{S}(X,\hat{\pi},\hat{Q}_{1-\alpha}(E_i))$, where $\hat{S}(X,\hat{\pi},\hat{Q}_{1-\alpha}(E_i))$ is the solution to the optimization problem in~\eqref{eq:def-opt-problem} with $S^- = \emptyset$ and $S^+ = \{1,\ldots,m\}$, while $\hat{Q}_{\tau}(E_i)$ is the $\lceil \tau (n+1) \rceil$ smallest value among $\{E_i\}$ for $i \in \{n+1,\ldots,2n\}$ for any $\tau \in (0,1]$.
The above simplification, combined with the assumed lack of randomization, will simplify our task considerably.

In order to keep the notation consistent, we will refer to $C^{\mathrm{oracle}}_{\tau}(X_{2n+1})$ as the optimal solution $S^{*}(X_{2n+1},f,\tau)$ to the oracle optimization problem in~\eqref{eq:oracle-interval}--\eqref{eq:oracle-interval-int}, where $f$ is the conditional probability density of $Y \mid X$.
Furthermore, without loss of generality, we divide the conformity scores $E_i$ of Algorithm~\ref{alg:sc} by $T$, so that they take values between 0 and 1 and can be directly interpreted as probabilities.

The proof will develop as follows.
\begin{enumerate} [label=(\roman*)]
\item Near-optimal length. First, we will prove in Lemma~\ref{lemma:length} that each interval $\hat{S}(X,\hat{\pi},\tau)$ typically cannot be much wider than the corresponding oracle interval $S^{*}(X,f,\tau + \delta \tau)$, for any fixed $\tau$ and an appropriately small $\delta \tau > 0$. This result will be based on Assumption~\ref{assumption:consistency} (consistency). Then, we will prove in Lemma~\ref{lemma:Q-hat} that $\hat{Q}_{1-\alpha}(\{E_i \}_{i \in \mathcal{D}^{\mathrm{cal}}})$ cannot be much larger than $1-\alpha$; this will be based on Assumption~\ref{assumption:consistency} (consistency) as well as on Assumption~\ref{assumption:iid} (i.i.d.~data), which makes the empirical quantiles to concentrate around their population values. Combining the above two lemmas will allow us to conclude that $\hat{S}(X,\hat{\pi},\hat{\tau})$ cannot typically be much wider than $S^{*}(X,f,1-\alpha)$.
\item Near-conditional length. First, we will prove in Lemma~\ref{lemma:Q-hat-2} that $\hat{Q}_{1-\alpha}(\{E_i \}_{i \in \mathcal{D}^{\mathrm{cal}}})$ cannot be much smaller than $1-\alpha$; again, this relies on the concentration of empirical quantiles due to the i.i.d.~assumption. Then, we will prove in Lemma~\ref{lemma:coverage} that the conditional coverage of $\hat{S}(X,\hat{\pi}, \tau)$ cannot be much smaller than $\tau$, for any fixed $\tau \in (0,1]$; this result relies on the consistency assumption.
Combining the above two lemmas will allow us to conclude that $\hat{S}(X,\hat{\pi},\hat{\tau})$ cannot typically have conditional coverage much smaller than $1-\alpha$.
\end{enumerate}
While Assumptions~\ref{assumption:iid}--\ref{assumption:consistency} will be critical, as previewed above, Assumptions~\ref{assumption:regularity}--\ref{assumption:smoothing} will play a subtler yet important role in connecting the various pieces.

\begin{lemma} \label{lemma:length}
Under Assumptions~\ref{assumption:iid}--\ref{assumption:smoothing},
for any $\tau \in (0,1)$ and $X \indep \mathcal{D}^{\mathrm{train}}$,
\begin{align*}
  \P{|\hat{S}(X,\hat{\pi},\tau)| \leq |S^{*}(X,f,\tau + 2 \eta_n^{1/3})| + 4 C \eta_n} \geq 1 - \delta_n,
\end{align*}
where $\delta_n \defeq \eta_n^{1/3} + \eta_n$.
\end{lemma}

\begin{lemma} \label{lemma:Q-hat}
For any $\tau \in (0,1]$, let $\hat{Q}_{\tau}(E_i)$ denote the $\lceil \tau (n+1) \rceil$ smallest value among the conformity scores $\{E_i\}$ for $i \in \mathcal{D}^{\mathrm{cal}}$, where $n = |\mathcal{D}^{\mathrm{cal}}|$ and
\begin{align*}
  E_i \defeq \min \left\{ \tau_t \in \{0, 1/T_n, \ldots, (T_n-1)/T_n, 1\} : Y_i \in \hat{S}(X_i, \hat{\pi}, \tau_t ) \right\}.
\end{align*}

Then, under Assumptions~\ref{assumption:iid}--\ref{assumption:smoothing}, for any $c>0$,
\begin{align*}
   \P{ \hat{Q}_{\tau}(E_i) \leq \tau + \epsilon_n } \geq 1 - 2 n^{-2c^2},
\end{align*}
where $\epsilon_n \defeq 3/n + 3 \eta_n^{1/3} + \eta_n + 2 c \sqrt{(\log n)/n}$.
\end{lemma}

\textbf{(i) Near-optimal length.}
Define $\delta_n \defeq \eta_n^{1/3} + \eta_n$ as in Lemma~\ref{lemma:length}, and $\epsilon_n \defeq 3/n + 3 \eta_n^{1/3} + \eta_n + 2 c \sqrt{(\log n)/n}$, for any $c>0$, as in Lemma~\ref{lemma:Q-hat}. In the event that $\hat{Q}_{1-\alpha}(E_i) \leq 1-\alpha + \epsilon_n$,
\begin{align*}
  & \P{|\hat{S}(X,\hat{\pi},\hat{Q}_{1-\alpha}(E_i))| \leq |S^{*}(X,f, 1-\alpha + \epsilon_n + 2 \eta_n^{1/3})| + 4C \eta_n} \\
  & \qquad \geq \P{|\hat{S}(X,\hat{\pi},1-\alpha + \epsilon_n)| \leq |S^{*}(X,f, 1-\alpha + \epsilon_n + 2 \eta_n^{1/3})| + 4C \eta_n} \\
  & \qquad \geq 1 - \delta_n,
\end{align*}
where the second inequality follows by applying Lemma~\ref{lemma:length} with $\tau = 1-\alpha + \epsilon_n$.
Further, as Lemma~\ref{lemma:Q-hat} tells us the above event occurs with high probability,
\begin{align*}
   \P{ \hat{Q}_{1-\alpha}(E_i) \leq 1-\alpha + \epsilon_n } \geq 1 - 2 n^{-2c^2},
\end{align*}
in general we have that
\begin{align*}
  & \P{|\hat{S}(X,\hat{\pi},\hat{Q}_{1-\alpha}(E_i))| \leq |S^{*}(X,f, 1-\alpha + \epsilon_n + 2 \eta_n^{1/3})| + 4 C \eta_n}
  \geq 1 - \delta_n - 2 n^{-2c^2}.
\end{align*}
By Assumption~\ref{assumption:regularity}, $f(y \mid x) > 1 / K$ for all $y \in [-C,C]$. This implies $|S^{*}(X,f, \tau)|$ is $K$-Lipschitz as a function of $\tau$. Therefore,
\begin{align*}
  & \P{|\hat{S}(X,\hat{\pi},\hat{Q}_{1-\alpha}(E_i))| \leq |S^{*}(X,f, 1-\alpha)| + 4 C \eta_n + K \left(\epsilon_n + 2 \eta_n^{1/3} \right)} \\
  & \qquad \geq \P{|\hat{S}(X,\hat{\pi},\hat{Q}_{1-\alpha}(E_i))| \leq |S^{*}(X,f, 1-\alpha + \epsilon_n + 2 \eta_n^{1/3})| + 4 C \eta_n} \\
  & \qquad \geq 1 - \delta_n - 2 n^{-2c^2}.
\end{align*}
Hence we have proved that
\begin{align*}
  \P{|\hat{S}(X,\hat{\pi},\hat{Q}_{1-\alpha}(E_i))| \leq |S^{*}(X,f, 1-\alpha)| + \gamma_n} \geq 1 - \xi_n,
\end{align*}
where $\gamma_n = 4 C \eta_n + K (\epsilon_n + 2 \eta_n^{1/3} )$ and $\xi_n = \delta_n + 2 n^{-2c^2}$.
For simplicity, we then set $c=1$.
This completes the proof of (i).

\begin{lemma} \label{lemma:Q-hat-2}
For any $\tau \in (0,1]$, let $\hat{Q}_{\tau}(E_i)$ denote the $\lceil \tau (n+1) \rceil$ smallest value among $\{E_i\}$ for $i \in \mathcal{D}^{\mathrm{cal}}$, where $n = |\mathcal{D}^{\mathrm{cal}}|$ and
\begin{align*}
  E_i \defeq \min \left\{ \tau_t \in \{0, 1/T_n, \ldots, (T_n-1)/T_n, 1\} : Y_i \in \hat{S}(X_i, \hat{\pi}, \tau_t ) \right\}.
\end{align*}
Then, under Assumptions~\ref{assumption:iid}--\ref{assumption:smoothing}, for any $c>0$,
\begin{align*}
   \P{ \hat{Q}_{\tau}(E_i) \geq \tau - \bar{\epsilon}_n } \geq 1 - 2 n^{-2c^2},
\end{align*}
where $\bar{\epsilon}_n \defeq 2/n + 3 \eta_n^{1/3} + (1 + 2 K) \eta_n + 2 c \sqrt{(\log n)/n}$.

\end{lemma}

\begin{lemma} \label{lemma:coverage}
Consider a test point $(X,Y) \indep \mathcal{D}^{\mathrm{train}}, \mathcal{D}^{\mathrm{cal}}$.
$\forall \tau \in (0,1]$, under Assumptions~\ref{assumption:iid}--\ref{assumption:smoothing},
\begin{align*}
  \P{ \P{Y \in \hat{S}(X,\hat{\pi}, \tau) \mid X} \geq \tau - 2 \eta_n^{1/3} } \geq 1 - \eta_n^{1/3} - \eta_n.
\end{align*}
\end{lemma}

\textbf{(ii) Near-conditional coverage.} 
Define $\bar{\epsilon}_n \defeq 2/n + 3 \eta_n^{1/3} + (1 + 2 K) \eta_n + 2 c \sqrt{(\log n)/n}$ as in Lemma~\ref{lemma:Q-hat-2}. Then, focus on the event
\begin{align*}
  \mathcal{E} \defeq \left\{ \hat{Q}_{1-\alpha}(E_i) \geq 1-\alpha - \bar{\epsilon}_n \right\}.
\end{align*}
In this event, for a new test point $(X,Y) \indep \mathcal{D}^{\mathrm{train}}, \mathcal{D}^{\mathrm{cal}}$,
\begin{align*}
  & \P{ \P{Y \in \hat{S}(X,\hat{\pi}, \hat{Q}_{1-\alpha}(E_i)) \mid X} \geq 1 - \alpha - \bar{\epsilon}_n - 2 \eta_n^{1/3} }  \\
  & \qquad \geq \P{ \P{Y \in \hat{S}(X,\hat{\pi}, 1-\alpha-\bar{\epsilon}_n) \mid X} \geq 1 - \alpha - \bar{\epsilon}_n - 2 \eta_n^{1/3} } \\
  & \qquad \geq 1 - \eta_n^{1/3} - \eta_n,
\end{align*}
where the last inequality follows by applying Lemma~\ref{lemma:coverage} with $\tau = 1 - \alpha - \bar{\epsilon}_n$.
Finally, note that Lemma~\ref{lemma:Q-hat-2} says the event $\mathcal{E}$ occurs with probability at least $1-2n^{-2}$, if we choose $c=1$.
Therefore,
\begin{align*}
  \P{ \P{Y \in \hat{S}(X,\hat{\pi}, \hat{Q}_{1-\alpha}(E_i)) \mid X} \geq 1 - \alpha - \bar{\epsilon}_n - 2 \eta_n^{1/3} }
  \geq 1 - \eta_n^{1/3} - \eta_n -2n^{-2}.
\end{align*}

\end{proof}

\subsubsection*{Proofs of technical lemmas}

The proofs of Lemmas~\ref{lemma:length}--\ref{lemma:coverage} will rely on the following additional lemma, which we state here and prove last.

\begin{lemma} \label{lemma:An}
Define the event $A_n$ as
\begin{align*}
  A_n \defeq \left\{ x : \sup_{j \in \{1,\ldots,m_n\}} | \hat{F}(b_j \mid x) - F(b_j \mid x) | > \eta_n^{1/3} \right\}.
\end{align*}
Then, under Assumptions~\ref{assumption:iid}--\ref{assumption:smoothing}, for any $X \indep \mathcal{D}^{\mathrm{train}}$,
\begin{align*}
  \P{X \in A_n}
  & \leq \eta_n^{1/3} + \eta_n.
\end{align*}
Furthermore, partitioning the calibration data points into
\begin{align*}
  & \mathcal{D}^{\mathrm{cal},a} \defeq \{ i \in \{n+1,\ldots,2n\} : X_i \in A_n\},
  & \mathcal{D}^{\mathrm{cal},b} \defeq \{ i \in \{n+1,\ldots,2n\} : X_i \in A_n^{\mathrm{c}}\},
\end{align*}
we have that, for any constant $c>0$,
\begin{align*}
  \P{|\mathcal{D}^{\mathrm{cal},a}| \geq n \left( \eta_n^{1/3} + \eta_n \right) + c\sqrt{n \log n} }
  & \leq n^{-2c^2}.
\end{align*}

\end{lemma}

\begin{proof}[\textbf{Proof of Lemma~\ref{lemma:length}}]

Consider the event $A_n$ defined in Lemma~\ref{lemma:An}, 
\begin{align*}
  A_n \defeq \left\{ x : \sup_{j \in \{1,\ldots,m_n\}} | \hat{F}(b_j \mid x) - F(b_j \mid x) | > \eta_n^{1/3} \right\},
\end{align*}
and let us restrict our attention to the case in which $X$ belongs to the complement of $A_n$.

Omitting the explicit dependence on $X$ and $\hat{\pi}$, we can write $\hat{S}(X,\hat{\pi},\tau) = [\hat{j}_1, \hat{j}_2]$, for some $\hat{j}_1, \hat{j}_2 \in \{1,\ldots,m_n\}$ such that $\hat{F}(b_{\hat{j}_2}) - \hat{F}(b_{\hat{j}_1-1}) \geq \tau$. Because we are assuming $X$ belongs to the complement of $A_n$, the triangle inequality implies
$F(b_{\hat{j}_2}) - F(b_{\hat{j}_1-1}) \geq \tau - 2 \eta_n^{1/3}$.
Consider now the oracle interval $S^{*}(X,f,\tau + 2 \eta_n^{1/3})$, which we can write in short as $[l^{*},u^{*}]$, for some $l^{*},u^{*} \in \mathbb{R}$ such that $F(u^{*})-F(l^{*}) \geq \tau + 2 \eta_n^{1/3}$.
Define now $j'_1,j'_2 \in \{1,\ldots,m_n\}$ as the indices of the discretized bins immediately below and above $l^{*},u^{*}$, respectively; precisely,
\begin{align*}
  j'_1 & \defeq \max \{ j \in \{1,\ldots,m_n\} : b_j < l^*\}, \\
  j'_2 & \defeq \min \{ j \in \{1,\ldots,m_n\} : b_j > u^*\}.
\end{align*}
This definition implies 
$$b_{j'_2}-b_{j'_1} \leq u^*-l^* + 4C/m_n,$$
as each bin has width $2C/m_n$. Furthermore, 
\begin{align*}
  \hat{F}(b_{j'_2}) - \hat{F}(b_{j'_1})
   & \geq \hat{F}(u^*) - \hat{F}(l^*) \\
   & \geq F(u^*) - F(l^*) - 2 \eta_n^{1/3} \\
  & \geq \tau.
\end{align*}
Above, the first inequality follows from the fact that $j'_1 < l^*$ and $j'_2 > u^*$, the second inequality follows from the assumption that $X$ belongs to the complement of $A_n$, and the third inequality follows directly from the definition of the oracle.
The result implies that $[j'_1,j'_2]$ would be a feasible solution for the discrete optimization problem solved by $\hat{S}(X,\hat{\pi},\tau)$; therefore, it must be the case that $\hat{j}_2 - \hat{j}_1 \leq j'_2 - j'_1$ because $\hat{j}_2 - \hat{j}_1$ is minimal among all feasible solutions to this problem. Therefore, we can conclude that, if $X$ belongs to the complement of $A_n$, then
\begin{align*}
  |\hat{S}(X,\hat{\pi},\tau)| 
  & = b_{\hat{j}_2}-b_{\hat{j}_1}
   \leq b_{j'_2}-b_{j'_1} \\
  & \leq |S^{*}(X,f,\tau + 2 \eta_n^{1/3})| + 4C/m_n.
\end{align*}
Finally, the proof is complete by applying Lemma~\ref{lemma:An}.

\end{proof}

\begin{proof}[\textbf{Proof of Lemma~\ref{lemma:Q-hat}}]

Take any $i \in \mathcal{D}^{\mathrm{cal},b}$, where $\mathcal{D}^{\mathrm{cal},b}$ is defined as in Lemma~\ref{lemma:An}:
\begin{align*}
  & \mathcal{D}^{\mathrm{cal},b} \defeq \{ i \in \{n+1,\ldots,2n\} : X_i \in A_n^{\mathrm{c}}\},
\end{align*}
where
\begin{align*}
  A_n \defeq \left\{ x : \sup_{j \in \{1,\ldots,m_n\}} | \hat{F}(b_j \mid x) - F(b_j \mid x) | > \eta_n^{1/3} \right\}.
\end{align*}

For any fixed $t \in \{0,\ldots,T_n\}$ and $\tau_t = t/T_n$, omitting the explicit dependence on $X$ and $\hat{\pi}$, we can write $\hat{S}(X,\hat{\pi},\tau_t) = [\hat{j}_1, \hat{j}_2]$, for some $\hat{j}_1, \hat{j}_2 \in \{1,\ldots,m_n\}$ such that $\hat{F}(b_{\hat{j}_2}) - \hat{F}(b_{\hat{j}_1-1}) \geq \tau_t$. 
Then, note that
\begin{align*}
  \P{E_i \leq \tau_t}
  & = \P{ Y_i \in \hat{S}(X_i, \hat{\pi}, \tau_t) } \\
  & = F(b_{\hat{j}_2}) - F(b_{\hat{j}_1-1}) \\
  & \geq \hat{F}(b_{\hat{j}_2}) - \hat{F}(b_{\hat{j}_1-1}) - 2 \eta_n^{1/3} \\
  & \geq \tau_t - 2 \eta_n^{1/3}.
\end{align*}
Above, the first inequality follows from the definition of $\mathcal{D}^{\mathrm{cal},b}$.
Equivalently, we can rewrite this as
\begin{align*}
  \P{E_i > \tau_t + 2 \eta_n^{1/3} + \delta_n } \leq 1 - \tau_t - \delta_n,
\end{align*}
for any $\delta_n > 0$.
Now, partition $\mathcal{D}^{\mathrm{cal},b}$ into the following two disjoint subsets:
\begin{align*}
  \mathcal{D}^{\mathrm{cal}, b1} & \defeq \{ i \in \mathcal{D}^{\mathrm{cal},b} : E_i \leq \tau_t + 2 \eta_n^{1/3} + \delta_n\}, \\
  \mathcal{D}^{\mathrm{cal}, b2} & \defeq \{ i \in \mathcal{D}^{\mathrm{cal},b} : E_i > \tau_t + 2 \eta_n^{1/3} + \delta_n\}.
\end{align*}
As in the proof of Lemma~\ref{lemma:An}, we bound $|\mathcal{D}^{\mathrm{cal},b2}|$ with Hoeffding's inequality. 
For any $i \in \mathcal{D}^{\mathrm{cal}}$, define $\tilde{E}_i = E_i$ if $i \in \mathcal{D}^{\mathrm{cal},b}$ and $E_i = \tau_t$ otherwise.
For any $\epsilon>0$,
\begin{align*}
  & \P{|\mathcal{D}^{\mathrm{cal}, b2}| \geq n (1-\tau_t - \delta_n ) + \epsilon } \\
  & \qquad \leq \P{\frac{1}{n} \sum_{i \in \mathcal{D}^{\mathrm{cal}, b}} \I{\tilde{E}_i > \tau_t + 2 \eta_n^{1/3} + \delta_n}  \geq \P{E_i > \tau_t + 2 \eta_n^{1/3} + \delta_n} + \frac{\epsilon}{n} } \\
  & \qquad = \P{\frac{1}{n} \sum_{i=1}^{n} \I{\tilde{E}_i > \tau_t + 2 \eta_n^{1/3} + \delta_n}  \geq \P{E_i > \tau_t + 2 \eta_n^{1/3} + \delta_n} + \frac{\epsilon}{n} } \\
  & \qquad \leq \P{\frac{1}{n} \sum_{i=1}^{n} \I{\tilde{E}_i > \tau_t + 2 \eta_n^{1/3} + \delta_n}  \geq \P{\tilde{E}_i > \tau_t + 2 \eta_n^{1/3} + \delta_n} + \frac{\epsilon}{n} } \\
  & \qquad \leq \exp \left( -\frac{2\epsilon^2}{n} \right).
\end{align*}
Therefore, setting $\epsilon = c\sqrt{n \log n}$, for some constant $c>0$, yields 
\begin{align*}
  \P{|\mathcal{D}^{\mathrm{cal}, b2}| \geq n (1-\tau_t - \delta_n) + c\sqrt{n \log n} } 
  & \leq n^{-2c^2}.
\end{align*}
As $|\mathcal{D}^{\mathrm{cal}, b1}| = n -|\mathcal{D}^{\mathrm{cal}, a}| - |\mathcal{D}^{\mathrm{cal}, b2}|$, combining the above result with that of Lemma~\ref{lemma:An} yields:
\begin{align*}
  \P{|\mathcal{D}^{\mathrm{cal}, b1}| \geq n \tau_t + n \delta_n - n\left( \eta_n^{1/3} + \eta_n \right) - 2 c \sqrt{n \log n}  } 
  & \geq 1 - 2 n^{-2c^2}.
\end{align*}
If we choose $\delta_n = \tau_t/n + \left( \eta_n^{1/3} + \eta_n \right) + 2 c \sqrt{(\log n)/n}$, this becomes
\begin{align*}
  \P{|\mathcal{D}^{\mathrm{cal}, b1}| \geq \tau_t (n+1)  } 
  & \geq 1 - 2 n^{-2c^2},
\end{align*}
which means
\begin{align*}
   \P{ \hat{Q}_{\tau_t}(E_i) \leq \tau_t + \tau_t/n + 3 \eta_n^{1/3} + \eta_n + 2 c \sqrt{(\log n)/n} } \geq 1 - 2 n^{-2c^2}.
\end{align*}
Now, consider any continuous $\tau \in (0,1]$, and let $t' = \min \{ t \in \{0,\ldots,T_n\} : \tau_t \geq \tau\}$. As $\tau_{t'} \geq \tau$, we know $\hat{Q}_{\tau_{t'}}(E_i) \geq \hat{Q}_{\tau}(E_i)$. Therefore, 
\begin{align*}
   & \P{ \hat{Q}_{\tau}(E_i) \leq \tau_{t'} + \tau_{t'}/n + 3 \eta_n^{1/3} + \eta_n + 2 c \sqrt{(\log n)/n} } \\
  & \qquad \geq \P{ \hat{Q}_{\tau_{t'}}(E_i) \leq \tau_{t'} + \tau_{t'}/n + 3 \eta_n^{1/3} + \eta_n + 2 c \sqrt{(\log n)/n} } \\
  & \qquad \geq 1 - 2 n^{-2c^2}.
\end{align*}
However, as $T_n = n$, we also have that $\tau_{t'} \leq \tau + 1/n$. Therefore,
\begin{align*}
   & \P{ \hat{Q}_{\tau}(E_i) \leq \tau + 1/n + \tau/n + 1/n^2 + 3 \eta_n^{1/3} + \eta_n + 2 c \sqrt{(\log n)/n} } \\
  & \qquad \geq 1 - 2 n^{-2c^2}.
\end{align*}
Finally, we simplify by replacing $1/n + \tau/n + 1/n^2$ with $3/n$, which preserves the inequality.

\end{proof}

\begin{proof}[\textbf{Proof of Lemma~\ref{lemma:Q-hat-2}}]

The proof is similar to that of the analogous upper bound in Lemma~\ref{lemma:Q-hat}.
Take any $i \in \mathcal{D}^{\mathrm{cal},b}$, where $\mathcal{D}^{\mathrm{cal},b}$ is defined as in Lemma~\ref{lemma:An}:
\begin{align*}
  & \mathcal{D}^{\mathrm{cal},b} \defeq \{ i \in \{n+1,\ldots,2n\} : X_i \in A_n^{\mathrm{c}}\},
\end{align*}
with
\begin{align*}
  A_n \defeq \left\{ x : \sup_{j \in \{1,\ldots,m_n\}} | \hat{F}(b_j \mid x) - F(b_j \mid x) | > \eta_n^{1/3} \right\}.
\end{align*}
For any $t \in \{0,\ldots,T_n\}$ and $\tau_t = t/T_n$, omitting the explicit dependence on $X$ and $\hat{\pi}$, we can write $\hat{S}(X,\hat{\pi},\tau_t) = [\hat{j}_1, \hat{j}_2]$, for some $\hat{j}_1, \hat{j}_2 \in \{1,\ldots,m_n\}$ such that $\hat{F}(b_{\hat{j}_2}) - \hat{F}(b_{\hat{j}_1-1}) \geq \tau_t$. 
Then,
\begin{align*}
  \P{E_i \leq \tau_t}
  & = \P{ Y_i \in \hat{S}(X_i, \hat{\pi}, \tau_t) } \\
  & = F(b_{\hat{j}_2}) - F(b_{\hat{j}_1-1}) \\
  & \leq \hat{F}(b_{\hat{j}_2}) - \hat{F}(b_{\hat{j}_1-1}) + 2 \eta_n^{1/3} \\
  & \leq \tau_t + 2 K \eta_n + 2 \eta_n^{1/3}.
\end{align*}
Above, the first inequality follows directly from the definition of $\mathcal{D}^{\mathrm{cal},b}$.
The second inequality follows from the observation that $\hat{S}(X_i, \hat{\pi}, \tau_t)$ could not be optimal if $\hat{F}(b_{\hat{j}_2}) - \hat{F}(b_{\hat{j}_1-1}) \geq \tau_t + 2 K \eta_n$ because it would be possible to obtain a shorter feasible interval by removing either the leftmost or the rightmost bin. In fact, each bin $j$ carries an estimated mass $\hat{\pi}_j \leq K \eta_n$, and $\hat{\pi}$ is assumed to be unimodal.
Fix any $\delta_n > 0$, and let us rewrite the above result as
\begin{align*}
  \P{E_i \leq \tau_t - 2 K \eta_n - 2 \eta_n^{1/3} - \delta_n } \leq \tau_t + \delta_n.
\end{align*}

Now, partition $\mathcal{D}^{\mathrm{cal},b}$ into the following two disjoint subsets:
\begin{align*}
  \mathcal{D}^{\mathrm{cal}, b1} & \defeq \{ i \in \mathcal{D}^{\mathrm{cal},b} : E_i \leq \tau - 2 K \eta_n - 2 \eta_n^{1/3} - \delta_n \}, \\
  \mathcal{D}^{\mathrm{cal}, b2} & \defeq \{ i \in \mathcal{D}^{\mathrm{cal},b} : E_i > \tau - 2 K \eta_n - 2 \eta_n^{1/3} - \delta_n \}.
\end{align*}
As in the proof of Lemma~\ref{lemma:Q-hat}, we will bound $|\mathcal{D}^{\mathrm{cal},b2}|$ with Hoeffding's inequality. 
For any $i \in \mathcal{D}^{\mathrm{cal}}$, define $\tilde{E}_i = E_i$ if $i \in \mathcal{D}^{\mathrm{cal},b}$ and $E_i = \tau_t$ otherwise.
For any $\epsilon>0$,
\begin{align*}
  & \P{|\mathcal{D}^{\mathrm{cal}, b1}| \geq n (1-\tau_t - \delta_n ) + \epsilon } \\
  & \qquad \leq \P{\frac{1}{n} \sum_{i \in \mathcal{D}^{\mathrm{cal}, b}} \I{\tilde{E}_i \leq \tau_t - 2 K \eta_n -  2 \eta_n^{1/3} - \delta_n}  \geq \P{E_i \leq \tau_t - 2 K \eta_n -  2 \eta_n^{1/3} - \delta_n} + \frac{\epsilon}{n} } \\
  & \qquad = \P{\frac{1}{n} \sum_{i=1}^{n} \I{\tilde{E}_i \leq \tau_t - 2 K \eta_n -  2 \eta_n^{1/3} - \delta_n}  \geq \P{E_i \leq \tau_t - 2 K \eta_n -  2 \eta_n^{1/3} - \delta_n} + \frac{\epsilon}{n} } \\
  & \qquad \leq \P{\frac{1}{n} \sum_{i=1}^{n} \I{\tilde{E}_i \leq \tau_t - 2 K \eta_n -  2 \eta_n^{1/3} - \delta_n}  \geq \P{\tilde{E}_i \leq \tau_t - 2 K \eta_n -  2 \eta_n^{1/3} - \delta_n} + \frac{\epsilon}{n} } \\
  & \qquad \leq \exp \left( -\frac{2\epsilon^2}{n} \right).
\end{align*}
Therefore, setting $\epsilon = c\sqrt{n \log n}$, for some constant $c>0$, yields 
\begin{align*}
  \P{|\mathcal{D}^{\mathrm{cal}, b1}| \geq n (1-\tau_t - \delta_n) + c\sqrt{n \log n} } 
  & \leq n^{-2c^2}.
\end{align*}
As $|\mathcal{D}^{\mathrm{cal}, b2}| = n -|\mathcal{D}^{\mathrm{cal}, a}| - |\mathcal{D}^{\mathrm{cal}, b1}|$, combining the above result with that of Lemma~\ref{lemma:An} yields:
\begin{align*}
  \P{|\mathcal{D}^{\mathrm{cal}, b2}| \geq n \tau_t + n \delta_n - n\left( \eta_n^{1/3} + \eta_n \right) - 2 c \sqrt{n \log n}  } 
  & \geq 1 - 2 n^{-2c^2}.
\end{align*}
If we choose $\delta_n = \tau_t/n + \left( \eta_n^{1/3} + \eta_n \right) + 2 c \sqrt{(\log n)/n}$, this becomes
\begin{align*}
  \P{|\mathcal{D}^{\mathrm{cal}, b2}| \geq \tau_t (n+1)  } 
  & \geq 1 - 2 n^{-2c^2},
\end{align*}
which means
\begin{align*}
   \P{ \hat{Q}_{\tau_t}(E_i) \geq \tau_t - \tau_t/n - 3 \eta_n^{1/3} - (1+2 K) \eta_n - 2 c \sqrt{(\log n)/n} } \geq 1 - 2 n^{-2c^2}.
\end{align*}

Now, consider any continuous $\tau \in (0,1]$, and let $t' = \max \{ t \in \{0,\ldots,T_n\} : \tau_t \leq \tau\}$. As $\tau \geq \tau_{t'}$, we know $\hat{Q}_{\tau}(E_i) \geq \hat{Q}_{\tau_{t'}}(E_i)$. Therefore, 
\begin{align*}
   & \P{ \hat{Q}_{\tau}(E_i) \geq \tau_{t'} - \tau_{t'}/n - 3 \eta_n^{1/3} - (1+2K) \eta_n - 2 c \sqrt{(\log n)/n} } \\
  & \qquad \geq    \P{ \hat{Q}_{\tau_t}(E_i) \geq \tau_t - \tau_t/n - 3 \eta_n^{1/3} - (1+2 K) \eta_n - 2 c \sqrt{(\log n)/n} } \\
  & \qquad \geq 1 - 2 n^{-2c^2}.
\end{align*}
However, as $T_n = n$, we also have that $\tau_{t'} \geq \tau - 1/n$. Therefore,
\begin{align*}
  1 - 2 n^{-2c^2}
  & \leq \P{ \hat{Q}_{\tau}(E_i) \geq (\tau-1/n) (1-1/n) - 3 \eta_n^{1/3} - (1+2K) \eta_n - 2 c \sqrt{(\log n)/n} } \\
  & \leq \P{ \hat{Q}_{\tau}(E_i) \geq \tau-\tau/n  - 1/n + 1/n^2 - 3 \eta_n^{1/3} - (1+2K) \eta_n - 2 c \sqrt{(\log n)/n} } \\
\end{align*}
Finally, we simplify by replacing $-1/n - \tau/n + 1/n^2$ with $-2/n$, which preserves the inequality.

\end{proof}

\begin{proof}[\textbf{Proof of Lemma~\ref{lemma:coverage}}]
Let us begin by conditioning on $X=x$, assuming $x \in A_n^{\mathrm{c}}$, where $A_n$ is defined as in Lemma~\ref{lemma:An}:
\begin{align*}
  A_n \defeq \left\{ x : \sup_{j \in \{1,\ldots,m_n\}} | \hat{F}(b_j \mid x) - F(b_j \mid x) | > \eta_n^{1/3} \right\}.
\end{align*}

Omitting the explicit dependence on $x$ and $\hat{\pi}$, we can write $\hat{S}(x,\hat{\pi},\tau) = [\hat{j}_1, \hat{j}_2]$, for some $\hat{j}_1, \hat{j}_2 \in \{1,\ldots,m_n\}$ such that $\hat{F}(b_{\hat{j}_2}) - \hat{F}(b_{\hat{j}_1-1}) \geq \tau$. 
\begin{align*}
  \P{Y \in \hat{S}(x,\hat{\pi}, \tau) \mid X=x}
  & = F(b_{\hat{j}_2}) - F(b_{\hat{j}_1-1}) \\
  & \geq \hat{F}(b_{\hat{j}_2}) - \hat{F}(b_{\hat{j}_1-1}) - 2 \eta_n^{1/3} \\
  & \geq \tau - 2 \eta_n^{1/3},
\end{align*}
where the second inequality follows from the definition of the complement of $A_n$.
Finally, we know from Lemma~\ref{lemma:An} that $\P{X \in A_n^{\mathrm{c}}} \geq 1-\eta_n^{1/3} - \eta_n$.
\end{proof}

\begin{proof}[\textbf{Proof of Lemma~\ref{lemma:An}}]
The first part of this result follows from the definition of $A_n$ by a union bound. For any fixed $j \in \{1,\ldots,m_n\}$,
\begin{align*}
  \P{X \in A_n}
  & = \P{\sup_{j' \in \{1,\ldots,m_n\}} | \hat{F}(b_{j'} \mid x) - F(b_{j'} \mid x) |^2 > \eta_n^{2/3}} \\
  & \leq m_n \, \P{| \hat{F}(b_j \mid x) - F(b_j \mid x) |^2 > \eta_n^{2/3} } \\
  & \leq m_n \left( \eta_n^{-2/3} \E{ \E{| \hat{F}(b_j \mid x) - F(b_j \mid x) |^2 \mid \mathcal{D}^{\mathrm{train}}} } \right) \\
  & \leq m_n \left( \eta_n^{-2/3} \eta_n^2 + \eta_n^2 \right) \\
  & \leq m_n \left( \eta_n^{4/3} + \eta_n^2 \right) \\
  & \leq \eta_n^{1/3} + \eta_n.
\end{align*}
The second inequality above is Markov's inequality, while the third inequality follows directly from Assumption~\ref{assumption:consistency}. The last inequality is a consequence of $m_n = \lfloor \eta_n^{-1} \rfloor$, also from Assumption~\ref{assumption:consistency}.

The second part of this result follows from Hoeffding's inequality. 
As we know from the above that $\P{X \in A_n} \leq \eta_n^{1/3} + \eta_n$, for any $\epsilon>0$,
\begin{align*}
  \P{|\mathcal{D}^{\mathrm{cal},a}| \geq n \left( \eta_n^{1/3} + \eta_n \right) + \epsilon }
  & \leq \P{|\mathcal{D}^{\mathrm{cal},a}| \geq n \P{X \in A_n} + \epsilon } \\
  & \leq \P{\frac{1}{n} \sum_{i=n+1}^{2n} \I{X_i \in A_n}  \geq \P{X_i \in A_n} + \frac{\epsilon}{n} } \\
  & \leq \exp \left( -\frac{2\epsilon^2}{n} \right).
\end{align*}
Therefore, setting $\epsilon = c\sqrt{n \log n}$, for some constant $c>0$, yields 
\begin{align*}
  \P{|\mathcal{D}^{\mathrm{cal},a}| \geq n \left( \eta_n^{1/3} + \eta_n \right) + c\sqrt{n \log n} }
  & \leq n^{-2c^2}.
\end{align*}
\end{proof}

\section{Numerical experiments} \label{sec:supp-data}

This section provides additional details about the numerical experiments with synthetic and real data.

\subsection{Base machine learning models}

We estimate the distribution of $Y \mid X$ using the following quantile regression models.
\begin{itemize}
    \item \textbf{Deep neural network}. The network is composed of three fully connected layers with a hidden dimension of 64, and ReLU activation functions. We use the pinball loss~\cite{taylor2000quantile} to estimate the conditional quantiles, with a dropout regularization of rate 0.1. The network is optimized using Adam~\cite{kingma2014adam} with a learning rate equal to 0.0005. We tune the optimal number of epochs by cross validation, minimizing the loss function on the hold-out data points; the maximal number of epochs is set to 2000.
    
    \item \textbf{Random forest}. We use the Python \texttt{Scikit-garden} implementation of quantile regression forests~\cite{meinshausen2006quantile}. We adopt the default hyper-parameters, except for the minimum number of samples required to split an internal node, which we set to 50, and the total number of trees, which we fix equal to 100.
\end{itemize}

Our numerical experiments were conducted on Xeon-2640 CPUs in a computing cluster. Each data set was analyzed using a single core and less than 5 GB of memory; the longest job took less than 12 hours. The computational cost of the novel part of CHR is negligible: the majority of the computing resources were dedicated to training the base models. 

\subsection{Additional experiments with synthetic data} \label{sec:supp-exp}

\begin{figure}[!htb]
  \centering
    \includegraphics[width=\textwidth]{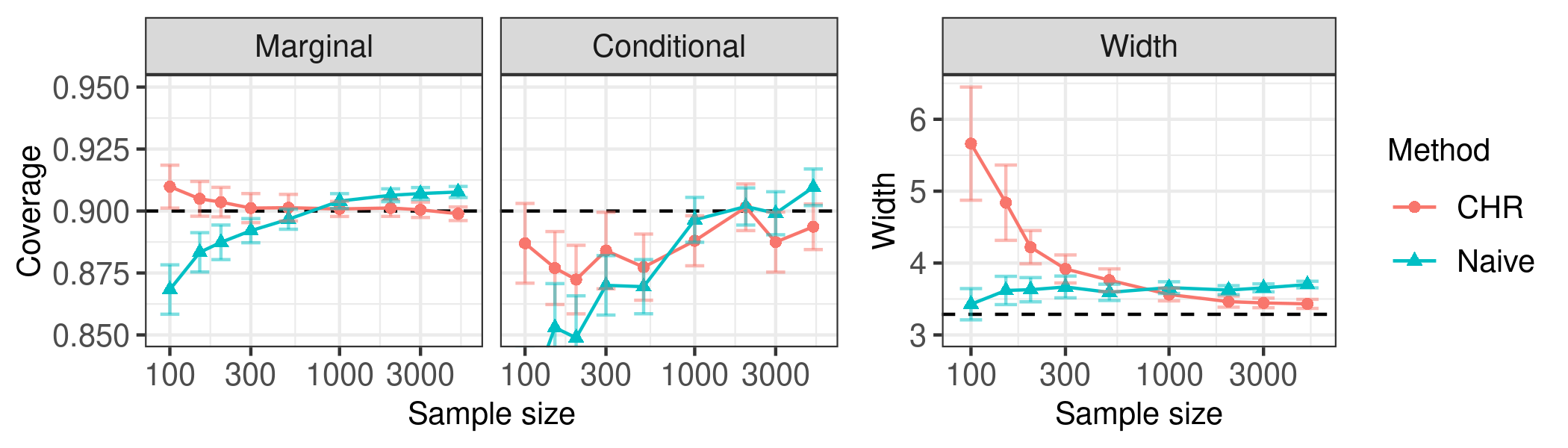}
\caption{Performance of our method (CHR) compared to that of naive uncalibrated prediction intervals based on the same deep neural network regression model, in the experiments of Figure~\ref{fig:research-example-1-boxes}.}
  \label{fig:research-example-1-boxes-naive}
\end{figure}

\subsection{Additional experiments with real data}

In Section~\ref{sec:real-data} of the main article, we compared the performance of our method to that of several benchmarks using a deep neural network base model. Figure~\ref{fig:results_real_rf} provides additional comparisons using a random forest base model. The bottom panel of this figure shows the average interval length. Our method (CHR) significantly outperforms all benchmarks by this metric, as it consistently constructs shorter intervals. The top panel of Figure~\ref{fig:results_real_rf} compares these alternative methods in terms of their worst-slab conditional coverage~\cite{cauchois2020knowing}, which we estimate as in~\cite{romano2020classification}. All methods achieve high conditional coverage on most data sets, except for CHR which tends to slightly undercover in the case of the two Facebook data sets (fb1 and fb2). Lastly, we note that all methods achieve exact 90\% marginal coverage, as guaranteed theoretically; see Table~\ref{tab:results_real} for additional performance details.

\begin{figure}[!htb]
  \centering
    \includegraphics[width=\textwidth]{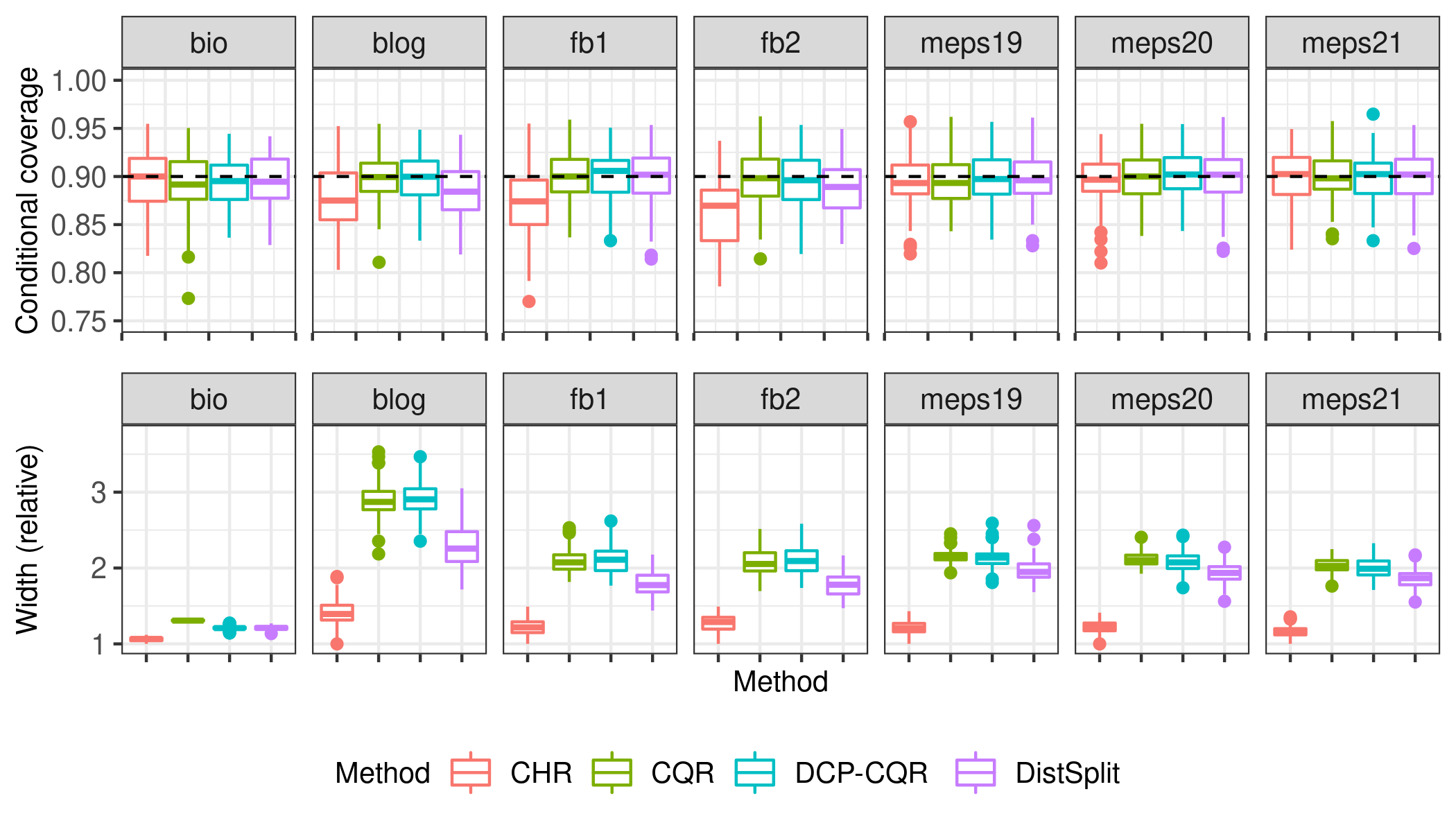}
    \caption{Performance of our method and benchmarks on several real data sets, using a random forest base model. Other details are as in Figure~\ref{fig:results_real_nnet}.} 
  \label{fig:results_real_rf}
\end{figure}

\begin{table}
\centering
\caption{Performance of our method and benchmarks on several real data sets, using either a deep neural network or a random forest base model. The numerical values indicate values averaged over 100 random test sets (standard deviations are in parenthesis). Other details are as in Figures~\ref{fig:results_real_nnet} and~\ref{fig:results_real_rf}.}
  \label{tab:results_real}
{\footnotesize
  
\begin{tabular}[t]{cccccccc}
\toprule
\multicolumn{2}{c}{ } & \multicolumn{3}{c}{Neural Network} & \multicolumn{3}{c}{Random Forest} \\
\cmidrule(l{3pt}r{3pt}){3-5} \cmidrule(l{3pt}r{3pt}){6-8}
\multicolumn{2}{c}{ } & \multicolumn{2}{c}{Coverage} & \multicolumn{1}{c}{ } & \multicolumn{2}{c}{Coverage} & \multicolumn{1}{c}{ } \\
\cmidrule(l{3pt}r{3pt}){3-4} \cmidrule(l{3pt}r{3pt}){6-7}
Data & Method & Marginal & Condit. & Width & Marginal & Condit. & Width\\
\midrule
 & CHR & 0.90 (0.01) & 0.88 (0.03) & 13.1 (0.3) & 0.90 (0.01) & 0.90 (0.03) & 10.4 (0.3)\\

 & CQR & 0.90 (0.01) & 0.88 (0.03) & 14.5 (0.2) & 0.90 (0.01) & 0.89 (0.03) & 12.9 (0.1)\\

 & DCP & 0.90 (0.01) & 0.88 (0.03) & 14.6 (0.3) & 0.90 (0.01) & 0.90 (0.02) & 11.7 (0.2)\\

 & DCP-CQR & 0.90 (0.01) & 0.87 (0.03) & 14.8 (0.4) & 0.90 (0.01) & 0.89 (0.03) & 11.9 (0.3)\\

\multirow[t]{-5}{*}{\centering\arraybackslash bio} & DistSplit & 0.90 (0.01) & 0.88 (0.03) & 14.7 (0.3) & 0.90 (0.01) & 0.90 (0.03) & 11.9 (0.3)\\
\cmidrule{1-8}
 & CHR & 0.90 (0.01) & 0.88 (0.03) & 10.9 (1.2) & 0.90 (0.01) & 0.88 (0.03) & 10.3 (1.2)\\

 & CQR & 0.90 (0.01) & 0.87 (0.04) & 15.0 (1.5) & 0.90 (0.01) & 0.90 (0.02) & 21.1 (1.6)\\

 & DCP & 0.90 (0.01) & 0.89 (0.03) & 1422.3 (0.1) & 0.90 (0.01) & 0.90 (0.03) & 1421.3 (0.1)\\

 & DCP-CQR & 0.90 (0.01) & 0.86 (0.04) & 14.0 (1.4) & 0.90 (0.01) & 0.90 (0.03) & 21.4 (1.8)\\

\multirow[t]{-5}{*}{\centering\arraybackslash blog} & DistSplit & 0.90 (0.01) & 0.87 (0.04) & 15.8 (1.6) & 0.90 (0.01) & 0.89 (0.03) & 16.7 (1.8)\\
\cmidrule{1-8}
 & CHR & 0.90 (0.01) & 0.87 (0.04) & 10.6 (0.9) & 0.90 (0.01) & 0.87 (0.04) & 11.2 (0.9)\\

 & CQR & 0.90 (0.01) & 0.89 (0.03) & 14.6 (1.0) & 0.90 (0.01) & 0.90 (0.02) & 19.2 (1.5)\\

 & DCP & 0.90 (0.01) & 0.90 (0.03) & 1303.3 (0.1) & 0.90 (0.01) & 0.90 (0.03) & 1302.6 (0.1)\\

 & DCP-CQR & 0.90 (0.01) & 0.89 (0.03) & 13.2 (1.1) & 0.90 (0.01) & 0.90 (0.03) & 19.4 (1.7)\\

\multirow[t]{-5}{*}{\centering\arraybackslash fb1} & DistSplit & 0.90 (0.01) & 0.89 (0.03) & 14.3 (1.1) & 0.90 (0.01) & 0.90 (0.03) & 16.5 (1.3)\\
\cmidrule{1-8}
 & CHR & 0.90 (0.01) & 0.87 (0.03) & 11.0 (0.9) & 0.90 (0.01) & 0.86 (0.03) & 10.8 (0.9)\\

 & CQR & 0.90 (0.01) & 0.89 (0.03) & 14.2 (0.9) & 0.90 (0.01) & 0.90 (0.03) & 17.7 (1.4)\\

 & DCP & 0.90 (0.01) & 0.90 (0.03) & 1964.0 (0.1) & 0.90 (0.01) & 0.89 (0.03) & 1963.4 (0.1)\\

 & DCP-CQR & 0.90 (0.01) & 0.89 (0.03) & 12.8 (1.1) & 0.90 (0.01) & 0.89 (0.03) & 17.8 (1.6)\\

\multirow[t]{-5}{*}{\centering\arraybackslash fb2} & DistSplit & 0.90 (0.01) & 0.89 (0.03) & 14.2 (1.1) & 0.90 (0.01) & 0.89 (0.03) & 15.1 (1.3)\\
\cmidrule{1-8}
 & CHR & 0.90 (0.01) & 0.90 (0.02) & 20.1 (1.3) & 0.90 (0.01) & 0.89 (0.03) & 18.4 (1.3)\\

 & CQR & 0.90 (0.01) & 0.89 (0.03) & 29.3 (1.2) & 0.90 (0.01) & 0.90 (0.02) & 32.6 (1.3)\\

 & DCP & 0.90 (0.01) & 0.89 (0.03) & 559.3 (0.0) & 0.90 (0.01) & 0.89 (0.03) & 559.0 (0.0)\\

 & DCP-CQR & 0.90 (0.01) & 0.89 (0.03) & 33.3 (2.3) & 0.90 (0.01) & 0.90 (0.03) & 32.2 (2.0)\\

\multirow[t]{-5}{*}{\centering\arraybackslash meps19} & DistSplit & 0.90 (0.01) & 0.90 (0.03) & 30.0 (2.3) & 0.90 (0.01) & 0.90 (0.03) & 29.8 (2.2)\\
\cmidrule{1-8}
 & CHR & 0.90 (0.01) & 0.90 (0.02) & 19.1 (1.2) & 0.90 (0.01) & 0.90 (0.02) & 17.7 (1.1)\\

 & CQR & 0.90 (0.01) & 0.88 (0.02) & 28.1 (1.0) & 0.90 (0.01) & 0.90 (0.03) & 30.5 (1.3)\\

 & DCP & 0.90 (0.01) & 0.89 (0.03) & 520.3 (0.0) & 0.90 (0.01) & 0.89 (0.03) & 520.1 (0.0)\\

 & DCP-CQR & 0.90 (0.01) & 0.89 (0.02) & 32.1 (2.2) & 0.90 (0.01) & 0.90 (0.02) & 29.9 (2.0)\\

\multirow[t]{-5}{*}{\centering\arraybackslash meps20} & DistSplit & 0.90 (0.01) & 0.89 (0.03) & 28.8 (2.0) & 0.90 (0.01) & 0.90 (0.03) & 27.9 (2.0)\\
\cmidrule{1-8}
 & CHR & 0.90 (0.01) & 0.90 (0.03) & 20.5 (1.2) & 0.90 (0.01) & 0.90 (0.03) & 19.2 (1.1)\\

 & CQR & 0.90 (0.01) & 0.89 (0.03) & 30.1 (1.3) & 0.90 (0.01) & 0.90 (0.02) & 33.4 (1.4)\\

 & DCP & 0.90 (0.01) & 0.89 (0.03) & 531.3 (0.0) & 0.90 (0.01) & 0.89 (0.03) & 531.0 (0.0)\\

 & DCP-CQR & 0.90 (0.01) & 0.89 (0.03) & 34.5 (2.4) & 0.90 (0.01) & 0.90 (0.02) & 32.9 (2.1)\\

\multirow[t]{-5}{*}{\centering\arraybackslash meps21} & DistSplit & 0.90 (0.01) & 0.90 (0.03) & 30.5 (2.0) & 0.90 (0.01) & 0.90 (0.03) & 30.6 (2.1)\\
\bottomrule
\end{tabular}

}
\end{table}

\begin{figure}[!htb]
  \centering
    \includegraphics[width=\textwidth]{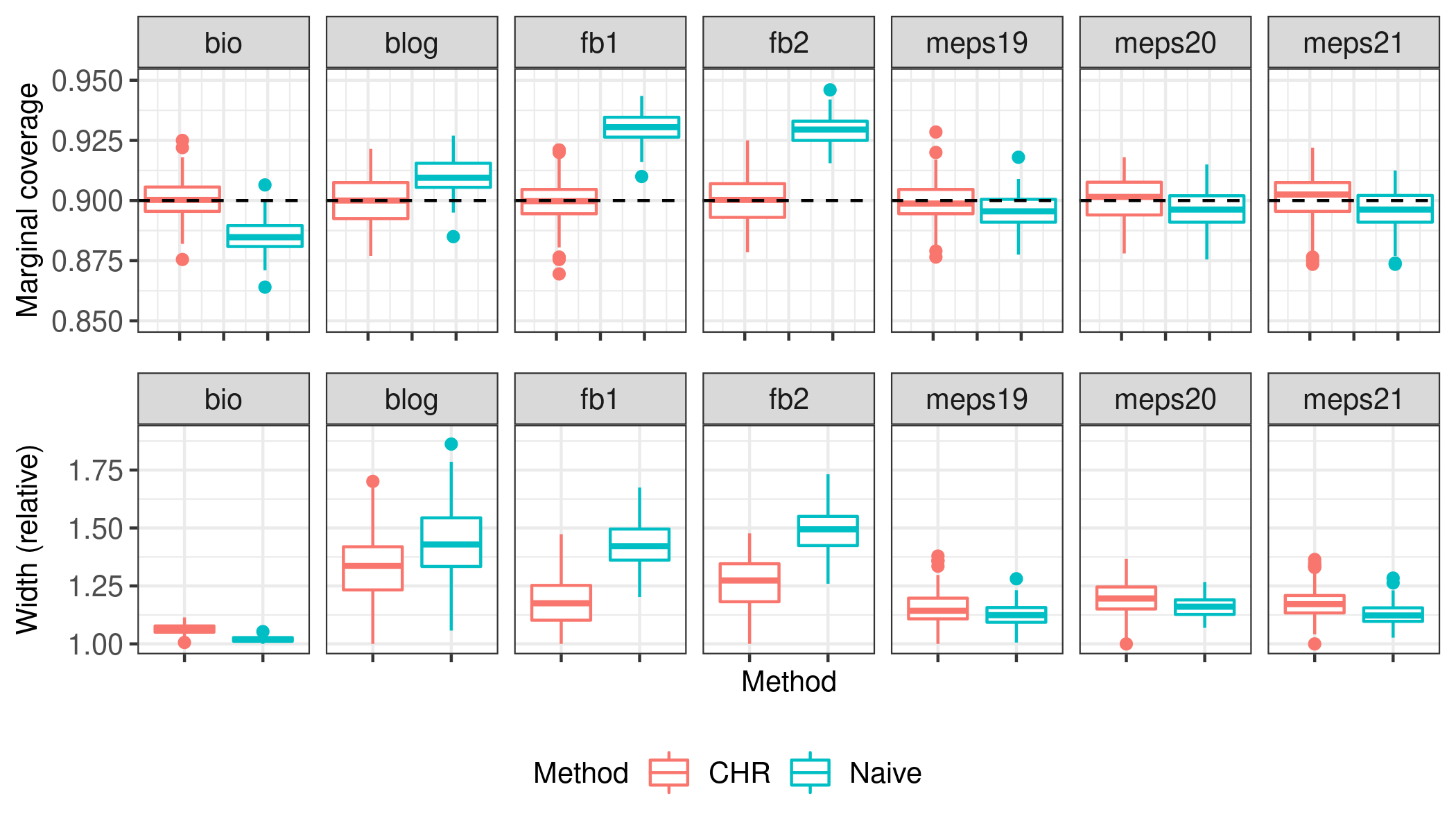}
\caption{Performance of our method (CHR) compared to that of naive uncalibrated prediction intervals based on the same deep neural network regression model, in the experiments with real data of Figure~\ref{fig:results_real_nnet}. Note that the top part of this plot shows marginal coverage.}
  \label{fig:results_real_nnet_naive}
\end{figure}

\end{document}